\documentclass[11pt,reqno]{amsart}

\usepackage{amsmath, amssymb, amsthm, amsfonts, cite, enumitem}

\usepackage[margin=25mm]{geometry}
\usepackage{hyperref}
\usepackage{graphics,pstricks}
\newtheorem{theorem}{Theorem}[section]
\newtheorem{proposition}[theorem]{Proposition}

\newtheorem{lemma}[theorem]{Lemma}
\newtheorem{remark}[theorem]{Remark}

\numberwithin{equation}{section}

\allowdisplaybreaks

\newcommand{\vn}{V_\omega}
\newcommand{\von}{V_{\omega,1}}
\newcommand{\vtn}{V_{\omega,2}}
\newcommand{\vin}{V_{\omega,i}}
\newcommand{\tvon}{\widetilde{V}_{L,1}}
\newcommand{\tvtn}{\widetilde{V}_{L,2}}

\newcommand{\drcf}{D} 
\newcommand{\drcnf}{D_{L}} 
\newcommand{\drcn}{D_{\omega,L}} 
\newcommand{\drcnr}{\mathcal{D}_L} 
\newcommand{\drcnl}{\mathbf{D}_L} 
\newcommand{\dos}{\mathcal{N}} 
\newcommand{\dosnf}{\mathcal{N}_{0,L}} 
\newcommand{\dosn}{\mathcal{N}_{\omega, L}} 
\newcommand{\spn}{\Upsilon_L}

\newcommand{\Phil}{\Phi^{\lambda}}

\newcommand{\Psil}{\Psi^{\lambda}}

\newcommand{\zl}{\zeta^{\lambda}}
\newcommand{\Rl}{R^{\lambda}}
\newcommand{\tl}{\theta^{\lambda}}
\newcommand{\vtl}{\vartheta^{\lambda}}
\newcommand{\Tl}{T^{\lambda}}

\newcommand{\ppru}{\mathcal{P}}

\newcommand{\vo}{V_{\omega,1}}
\newcommand{\vt}{V_{\omega,2}}

\newcommand{\deltap}{\delta^+}
\newcommand{\deltam}{\delta^-}

\DeclareMathOperator{\tr}{Tr}

\newcommand{\Z}{\mathbb{Z}}
\newcommand{\R}{\mathbb{R}}
\newcommand{\C}{\mathbb{C}}
\newcommand{\n}{\mathbb{N}}

\newcommand{\p}{\mathbb{P}}
\newcommand{\esp}{\mathbb{E}}

\newcommand{\h}{\mathcal{H}}

\newcommand{\bra}{\langle}
\newcommand{\ket}{\rangle}

\newcommand{\vp}{\phi}

\def \simless {\mathbin{\lower 3pt\hbox{$\rlap{\raise 5pt
              \hbox{$\char'074$}}\mathchar"7218$}}}

\author[G. R. Moreno Flores and A. Taarabt]{Gregorio R. Moreno$^{1}$ Flores and Amal Taarabt$^2$}
\date{}

\address{Facultad de Matem\'aticas\\
Pontificia Universidad Cat\'olica de Chile\\
Vicu\~na Mackenna 4860, Macul\\
Santiago, Chile}

\email{grmoreno@mat.uc.cl, amtaarabt@mat.uc.cl}


\thanks{ AMS 2010 {\it subject classifications}. 82B44, 47B80}

\thanks{$^{1,2}$Facultad de Matem\'aticas, Pontificia Universidad Cat\'olica de Chile.}

\title[Dirac Operators in a Decaying Potential]{One-dimensional Discrete Dirac Operators in a Decaying Random Potential II: Clock, Schr\"odinger  and Sine statistics}

\begin{document}

\begin{abstract}
	We consider one-dimensional discrete Dirac models in vanishing random environments. In a previous work \cite{BMT01}, we showed that these models exhibit a rich phase diagram in terms of their spectrum as a function of the rate of decay of the random potential. 
	
This article is devoted to their spectral statistics. We show that the rescaled spectrum converges to the clock process for fast decay and to the Schr\"odinger/Sine processes from random matrix theory for critical decay. This way, we recover all the regimes previously identified for the Anderson model in a similar context \cite{KVV}. Poisson statistics, which should appear in the model with slow decay, are left as an open problem.

The core of the proof consists in a suitable scaling limit for the Pr\"ufer phase and monotonicity arguments, yielding an alternative to the approach of \cite{KVV}. For one of the models, we also obtain the scaling limit of the Pr\"ufer radii and discuss the consequences for the limiting shape of the eigenfunctions. 
\end{abstract}

\maketitle

\tableofcontents


\section{Introduction}
The RAGE theorem is probably one of the most beautiful results in mathematical analysis. Roughly speaking, it states that the spectrum of a (random) operator encloses most of the information about the dynamics generated by its associated unitary group. Unfortunately, the precise nature of this (random) object is essentially intractable and its properties have to be derived by means of quite indirect techniques. In principle, a way to simplify the problem consists in restricting the operator to finite boxes. Focussing in the random case from now on, the spectrum then reduces to the eigenvalues of this finite dimensional random matrix and one hopes to recover information on the original infinite dimensional problem through a limit procedure.
At a microscopic scale, the spectrum of such a box restriction is an intricated point process with highly non-trivial correlations which, except for some classical random matrix ensembles, cannot be described explicitly. On the other hand, at a macroscopic scale, one may look at the limit of the empirical measure of these eigenvalues. This is the (integrated) density of states. Of course, lots of information is lost in this procedure. However, some non-trivial properties of the dynamics can still be inferred. For instance, Lifshits tails, i.e. the dramatic decay of the density of state at the edge of the spectrum, are known to imply dynamical localization by means of well-known techniques such as fractional moments or multiscale analysis (see \cite{AW} and references therein for references on random operators).

There is a middle path which consists in looking at mesoscopic scales i.e. to consider the distribution of the eigenvalues in energy intervals of size inversely proportional to the side-length of the box and take a large box limit. This is what is usually understood as spectral statistics. The limiting point processes obtained in this way usually fall in a few different categories, highliting the fact that the mesoscopic spectral statistics of random operators often do not depend on the precise nature of the model under consideration, a phenomenon known in the probability literature as universality. As such, the study of this regime is essentially probabilistic: it focusses on the probabilistic nature of the spectrum which is not related to the dynamical properties of the operators in a completely obvious way.
Spectral statistics have been studied in various works, mainly for discrete and continuum Anderson models \cite{GK14,Killip-Nakano,Min07,Molch82,Killip-Stoiciu,VV,W}.

If we drop the randomness, there is hope to perform explicit computations on the free operator and, as a rule of thumbs, deep inside the spectrum, the limiting point process simply consists in evenly spaced points. This is known as the clock (or picket fence) process. Interestingly enough, this deterministic behaviour remains when the randomness is not too relevant, as we will see in our main working example.
In a sense, this behaviour is completely deterministic. On the other hand, for highly disordered random operators, the limiting object consists of random points independently scattered on the line i.e. a Poisson point process. The heuristics behind this phenomenon is quite clear. The eigenfunctions of the box restriction of such operators are extremely localized, hence, eigenfunctions corresponding to different eigenvalues are essentially supported in disjoint boxes, which results is mostly independent eigenvalues. A careful analysis then allows to approximate the spectrum of our finite dimensional operator by the spectrum of the direct sum of restrictions to smaller boxes, each one of them contributing with at most a single eigenvalue in each rescaled interval. This is exactly the way Poisson statistics arise as the limit of triangular arrays of independent point processes.

Now, this is not the end of the story. Random operators in certain ‘critical regimes’ are known to exhibit spectral statistics which interpolate between the clock and the Poisson behavior \cite{Killip-Stoiciu,VV, KVV,  AllezDumaz}: eigenvalue repulsion is not as strong as in the deterministic case but is strong enough to give rise to point processes with non-trivial correlations in the limit. From another point of view, eigenfunctions corresponding to different eigenvalues overlap but in a very subtile way.
Two examples of this phenomenon are given by the Schrödinger and Sine processes which appear as the spectral statistics of a wide range of random operators and, appealing to the probabilistic terminology, are then universal \cite{VV}. 

The literature on this phenomenon has grown in the last years, especially in the context of one-dimensional Anderson models with decaying potentials \cite{KVV, KU, Nakano, Kotani-Nakano17} (see also \cite{Killip-Stoiciu} for a unitary counterpart). 
This model was initially introduced (in arbitrary dimensions) as an attempt to understand how the absolutely continuous spectrum can survive the addition of a random potential \cite{Krishna, Bou1, Bou2, Figo}. Even prior to the study of its spectral statistics, the one-dimensional case received particular attention and is know to yield a rich phase diagram at the spectral level (see \cite{KLS} for the Anderson model and \cite{deOliveiraPrado,BMT01} for the model studied in this work).

This Poisson/Sine dichotomy is conjectured to apply to a large class of quantum systems. This phenomenon is the object of the Berry-Tabor and Bohigas-Gianini-Smith conjectures. We refer the reader to \cite[Chapter 17]{AW} for a discussion of this matter which falls well beyond the scope of this article. 


\vspace{2ex}

In this work, we consider a one-dimensional discrete Dirac model with two sorts of decaying potentials. One where the strength of the disorder is sent to $0$ with the size of the system (Model I) and one where a decaying envelope modulates an i.i.d. potential (Model II). 
This complements our work on the spectral regimes of Model II \cite{BMT01}. 
We recover the scaling limits obtained previously for analogous Anderson models (see Theorem \ref{thm:main}). As anticipated, when the disorder vanishes rapidly enough, spectral statistics are governed by the clock process. On the contrary, when the strength of the disorder is critically tuned, we recover the Schrödinger/Sine processes. As in the Anderson model, the case of a slowly vanishing potential is left open, although we strongly believe that spectral statistics should be Poissonian. This regime falls out of the scope of standard techniques and the point of view adopted here. We further comment on this case in Remark \ref{rk:Poisson} below and defer this problem to future work (see \cite{Kotani-Nakano17} for a related continuum model).

At the technical level, \cite{KVV} considers the scaling limit of transfer matrices, from where the scaling limit of the rescaled spectrum follows. Due to the matrix nature of the Dirac model, such a scheme becomes rather unmanageable and it is more convenient to work with  the scaling limit of a certain Pr\"ufer transform, which is shown to converge to the SDE defining the Schr\"odinger/Sine processes respectively
(Theorems \ref{thm:scaling-model-I} and \ref{thm:scaling-model-II}). 
The convergence of the rescaled spectrum, which is the object of our main result, then follows from elementary arguments which simplify the approach of \cite{KVV} (Theorem \ref{thm:convergence-pp}). 
We note that, with our techniques, we obtain convergence to the clock and Schr\"odinger/Sine processes at once in their corresponding regimes. 
As a preliminary, we discuss the density of states of our models (see Section \ref{sec:intro-dos}).

Full proofs will be given for Model I while proofs for Model II will be only sketched. In this case, the convergence of the rescaled spectrum is a consequence of certain scaling limits which can be proved very similarly to Model I, supplemented by general arguments which will not be repeated here. We will refer the reader to \cite{KVV} for the missing details.


\section{Models and Results}

We set $\n^* = \{ 1,2,\cdots \}$ and consider the Hilbert space $\h=l^2(\n^*,\C^2)$ 
with the canonical basis $(\delta_n^\pm)_n$. A vector $\Phi=(\Phi_n)_n\in\h$ is given by two sequences $\vp^\pm=(\vp_n^\pm)_n\in\h$ such that
\begin{eqnarray*}
	\Phi_n
	=
	\begin{pmatrix}
		\vp^+_n \\ \vp^-_n
	\end{pmatrix}\quad \mathrm{for}\ n\in\n^*,
\end{eqnarray*}
which we will occasionally denote by $\Phi=\vp^+ \otimes \vp^-$.


\subsection{Dirac operators in random potentials}


We consider the free Dirac operator $D$ defined by
\begin{eqnarray}
	D
	=
	\begin{pmatrix}
		m & d \\
		d^* & -m
	\end{pmatrix}\quad\mathrm{on}\quad\h,
\end{eqnarray}
with mass $m\ge0$ and where $d$ and $d^*$ are the shift operators acting as $(du)_n=u_n-u_{n+1}$ and $(d^*u)_n=u_n-u_{n-1}$ for $u\in\ell^2(\n^*,\C)$ with the convention $u_0=0$.
Explicitly, if $\Phi=\vp^+\otimes\vp^-\in\h$, then
\begin{eqnarray*}
	D\Phi
	=
	\begin{pmatrix}
		m \vp^+ + d\vp^-
		\\
		d^*\vp^+ - m \vp^-
	\end{pmatrix}.
\end{eqnarray*}
The operator $D$ is self-adjoint and satisfies 
$$ D^2=\begin{pmatrix}
\Delta+m^2&0\\0&\Delta+m^2\end{pmatrix},
$$
where $\Delta$ is the discrete Laplacian defined on $\ell^2(\n^*,\C)$ by $\Delta u(n):=2u(n)-u(n+1)-u(n-1)$, with the convention $u(0)=u(-1)=0$. 
This yields that $\sigma(D^2)=[m^2+m^2+4]$ \cite{Golenia,COP11}, and hence the spectrum of the free Dirac operator $D$ is absolutely continuous and is given by
$$\Sigma:=\sigma(D)=\sigma_\mathrm{ac}(D)=[-\sqrt{m^2+4},-m]\cup[m,\sqrt{m^2+4}].$$

To define the perturbed operators we introduce a probability space $(\Omega,\mathbb{F},\p)$ and a family of bounded independent and identically distributed random variables $\{\omega_{n,i};\, n\in\n^*,\, i=1,\, 2\}$. 
Let $\esp$ be the expectation with respect to $\p$. We assume that $\esp[\omega]=0$ and $\esp[\omega^2]=1$.
We will consider two different models. In both cases, the potential $\vn$ will be given by a multiplication operator acting on the canonical basis as
\begin{eqnarray}
	&&
	\vn \delta^+_n = \von(n) \delta^+_n,
	\quad
	\vn \delta^-_n = \vtn(n) \delta^-_n,
\end{eqnarray}
where $(\von(n))_n$ and $(\vtn(n))_n$ are random sequences specific to each model.
\bigskip

Let $\drcnf$ be the restriction of $D$ to box $\Lambda_L=[0,L]$ with boundary conditions $\vp^-_0=\vp^+_L=0$. We then define
\begin{eqnarray}
	\drcn = \drcnf + \vn, \quad \text{on}\quad  \ell^2(\Lambda_L,\C^2),
\end{eqnarray}
where the local potential is defined by

\vspace{2ex}

\begin{itemize}
	\item {\noindent  \textsc{{Model I (Vanishing coupling):}} }
	$\displaystyle \vin(n)
	=
	\vin^L(n)
	=
	\gamma \frac{\omega_{n,i}}{L^{\alpha}},
	\quad
	i = 1,\, 2.$
	
	\vspace{2ex}
	
	\item {\noindent \textsc{{Model II (Decaying potential):}}} 
$\displaystyle \vin(n)
	=
	\gamma \frac{\omega_{n,i}}{n^{\alpha}},
	\quad
	i = 1,\, 2.$
\end{itemize}

%
%
%

\vspace{2ex}

Let us recall the spectral and dynamical properties of Model II.
In \cite{BMT01}, we have obtained deloca- \\lization/localization results with respect to the value of $\alpha$. Compared to the Anderson model, this operator yields to additional technical difficulties due to its order one and matrix nature. We recovered all the spectral regimes from [84], the three regimes being characterized by  $\alpha>\frac12$ (with a.c. spectrum), $\alpha=\frac12$ (with a spectral transition depending on $\gamma$ in this case) and $\alpha<\frac12$ (p.p. spectrum) respectively. At criticality, we found an elementary argument implying delocalization in the spirit of \cite{GKT}. Finally, we showed dynamical localization for slow decay using the fractional moments method. This allowed us to work with single site distributions outside the scope of the Kunz-Souillard method used in \cite{Si82}. Our proof can be adapted to the Anderson model with some simplifications and extends the results of \cite{Si82} to the whole range of application of the fractional moments method. In \cite{BMT02}, we proved dynamical localization for the continuum Anderson model in a slowly decaying random potential by means of the continuum fractional moments method of \cite{AENSS,HSS-continousFMM}. This was left as an open problem in \cite{DS-KunzSouillard} where the authors developed a continuum version of the Kunz-Souillard method.
  
We now proceed with the formulation of our results.


\subsection{The density of states}
\label{sec:intro-dos}
The integrated density of states $\dosn$ (IDS) is defined as the number of eigenvalues of $D_{\omega,L}$ per unit volume i.e.\begin{equation*}
\dosn(E):=\frac{\sharp\{\text{eigenvalues of}\ D_{\omega,L}\ \text{less than}\  E\}}{|\Lambda_L|},
\end{equation*}
or, equivalently,
\begin{equation*}
\dosn(E):=\frac{1}{|\Lambda_L|}\esp \left(\tr P_{\omega,L}(E)\right)
\end{equation*}
where $P_{\omega,L}(E)=\chi_{(-\infty,E]}(D_{\omega,L})$ is the Fermi projector associated to $D_{\omega,L}$. We then let
$$D_{\omega} = \lim_{L\to\infty} \dosn,$$
provided the limit exists. Note that, for Model II, this corresponds to the IDS of an infinite-volume operator. We denote the analogous objects corresponding to $D_L$ and $D$ by $\dosnf$ and $\dos_0$ respectively. 

Our first result in an explicit expression for $\dos_0$. It might be well-known to the specialist but, as we did not find it stated in the literature, we include it for the convenience of the reader.

\begin{proposition}
	The limit $\displaystyle\dos_0 = \lim_{L\to\infty} \dosnf$ exists and admits the density $(2\pi)^{-1}\rho$, where
	\begin{eqnarray*}
	\rho(E)
	=
	\frac{4E}{\sqrt{(E^2-m^2)(m^2+4-E^2)}},
	\quad
	m < |E| < \sqrt{m^2+4}.
\end{eqnarray*}
\end{proposition}
The limit for the random operators turns out to be the same. This fact was already noticed by  \cite{DolaiKrishna} for $d\geq 3$.
\begin{theorem}
	For both models and for all $\alpha>0$, the limit
	\begin{eqnarray*}
		\dos = \lim_{L\to\infty} \dos_{\omega,L},
	\end{eqnarray*}
	exists and, furthermore, $\dos=\dos_0$ almost-surely.
\end{theorem}


\subsection{Spectral statistics}


We fix $E\in(m, \sqrt{m^2+4})$ and look for eigenvalues of the form $E+\frac{\lambda}{\rho L}$ where $\rho=\rho(E)$. 
This is equivalent to determining the spectrum of the operator 
\begin{eqnarray*}
	\drcnr=\rho L(\drcn-E),
\end{eqnarray*}
with the specified boundary conditions $\vp^-_0=\vp^+_L=0$. We define $\spn$ as the set of the eigenvalues of $\drcnr$ i.e. the point process given by
\begin{eqnarray*}
	\spn 
	=
	\left\{
		\lambda\in\R:\, \det \left( \drcnr - \lambda \right)=0
	\right\}.
\end{eqnarray*}
This is our main object of study. We will be interested in its scaling limit as $L\to\infty$. Later, we will characterize $\spn$ in terms of a phase process which is amenable for scaling limits. But first, we need to introduce the three families of limiting processes we will encounter in this work.

\vspace{1ex}


\vspace{2ex}

\noindent \textsc{\underline{The clock process:}} For $\eta\in\R$, we define
\begin{eqnarray}
	\textbf{Clock}_{\eta}
	=
	\left\{
		2\pi k+\pi + 2\eta,\, k\in\Z
	\right\}.
\end{eqnarray}
This will correspond to the scaling limit of $\spn$ for both models and $\alpha>\frac12$.

\vspace{2ex}

\noindent \textsc{\underline{The Scr\"odinger process:}} Consider the family of coupled SDEs
\begin{eqnarray}
\label{eq:SDE-schrodinger}
		d\varphi^{\lambda}
		&=&
		\lambda dt
		+
		\text{Re}\left\{
			e^{-i\varphi^{\lambda}}dB
		\right\}
		+
		dW,
		\quad
		\varphi^{\lambda}(0)=0,
	\end{eqnarray}
	for $\lambda\in\R$ and $t\in[0,1]$ and where
	$B$ and $W$ are independent complex standard motion and one-dimensional standard Brownian motions respectively. From \cite{KVV}, Corollary 4, we know that there exists a unique strong solution for such a system and that, for each fixed $t$, the function $\lambda\mapsto\vartheta^{\lambda}(t)$ is almost-surely analytic and increasing. 
	For $\tau>0$, we define the Schr\"odinger point process
	\begin{eqnarray}\label{eq:Schodinger point process}
		\textbf{Sch}_{\tau}
		=
		\{\lambda:\, \varphi^{\lambda/\tau}(\tau) \in 2\pi\Z\}.
	\end{eqnarray}
	This will correspond to the scaling limit of $\spn$ for Model I and $\alpha=\frac12$.

\vspace{2ex}

\noindent \textsc{\underline{The Sine process:}} 
Let $\beta\in\R$ and consider the family of coupled SDEs
\begin{eqnarray}
\label{eq:SDE-sine}
		d\alpha^{\lambda}
		&=&
		\lambda \frac{\beta}{4}e^{-\frac{\beta}{4}t} dt
		+
		\text{Re}\left\{
			\left(
				e^{-i\alpha^{\lambda}}-1
			\right)\,
			dZ
		\right\},
		\quad
		\alpha^{\lambda}(0)=0,
	\end{eqnarray}
	for $\lambda\in\R$ and $t\geq 0$ and where $Z$ is a complex standard Brownian motion. Then, the function 
	$t\mapsto \lfloor \alpha^{\lambda}(t)/2\pi \rfloor$ is non-decreasing and the limit $\displaystyle\alpha_{\infty}(\lambda):=\lim_{t\to\infty} \alpha^{\lambda}(t)$ satisfies $\alpha_{\infty}(\lambda)\in 2\pi \Z$. The sine process in the interval $[\lambda_1,\lambda_2]$ is then defined as
	\begin{eqnarray}
		\textbf{Sine}_{\beta}[\lambda_1,\lambda_2]
		=
		\frac{\alpha_{\infty}(\lambda_2)-\alpha_{\infty}(\lambda_1)}{2\pi}.
	\end{eqnarray}
	This will correspond to the scaling limit of $\spn$ for Model II and $\alpha=\frac12$.
%



\vspace{1ex}

Let us introduce the proper notion of convergence for point processes. Note that a point process can be interpreted as a random point measure by considering its empirical distribution.
We say that a sequence of random measures $(\mu_n)_n$ on $(\mathbb{R},\mathcal{B}(\mathbb{R}))$ converges in law to $\mu$ in the vague convergence topology if
\begin{eqnarray*}
	\lim_{n\to\infty}
	\esp\left[
		e^{-\int f \, d\mu_n}
	\right]
	=
	\esp\left[
		e^{-\int f \, d\mu}
	\right],
\end{eqnarray*}
for all continuous non-negative compactly supported functions $f:\mathbb{R}\to\mathbb{R}$ \cite{Durrett}.

\bigskip

Let us now formulate our main result.  In the theorem below, we consider the positive component of the spectrum as the other case is similar. We also omit the energy level $E=\sqrt{m^2+2}$ which yields similar results but requires a slightly different treatment (see \cite[Theorem 7]{KVV}).
For $x\in\R$, let $[x]$ denote the unique value in $[0,\pi)$ such that $x=[x](\text{mod}\, 2\pi)$.

\begin{theorem}
\label{thm:main}
	Let $k\in(-\pi,-\frac{\pi}{2})$ with $k\neq -\frac{3\pi}{4}$, let $E\in(m,\sqrt{m^2+4})$ be the corresponding energy such that  $\cos k = - \frac{1}{2} \sqrt{E^2-m^2}$. Let $\sigma^2 = \frac{p_1(E)^2+p_2(E)^2}{\sin^2(2k)}$ and $\eta_L = (2L-1)k$.
	We have the following convergences in law in the vague convergence topology:
	\begin{enumerate}
		\item[1.-] \textsc{\underline{Model I}}
			\begin{enumerate}
				\item[1a.-] If $\alpha>\frac{1}{2}$, then 
				\begin{eqnarray*}
					\Upsilon_{\omega,L} - [2\eta_L]-\pi
					\,
					\longrightarrow
					\,
					\textbf{Clock}_{0}.
				\end{eqnarray*}
				
				\vspace{1ex}
		
				\item[1b.-] If $\alpha=\frac12$, then
				\begin{eqnarray*}
					\Upsilon_{\omega,L} - [2\eta_L]-\pi
					\quad
					\longrightarrow
					\quad
					\textbf{Sch}_{\sigma^2}.
				\end{eqnarray*}
			\end{enumerate}
		
		\item[2.-] \textsc{\underline{Model II}}
			\begin{enumerate}
				\item[2a.-] If $\alpha>\frac{1}{2}$, then 
				\begin{eqnarray*}
					\Upsilon_{\omega,L} - [2\eta_L]-\pi
					\,
					\longrightarrow
					\,
					\textbf{Clock}_{0}.
				\end{eqnarray*}
				
				\vspace{1ex}
				
				\item[2b.-] If $\alpha=\frac12$, then
				\begin{eqnarray*}
					\Upsilon_{\omega,L} - [2\eta_L]-\pi
					\quad
					\longrightarrow
					\quad
					\textbf{Sine}_{\beta},
				\end{eqnarray*}
				with $\beta=\tfrac{\! 2}{\,\sigma^2}$.
			\end{enumerate}
	\end{enumerate}
\end{theorem}

This will be proved by considering the scaling limit of a certain Pr\"ufer transform (see Theorem \ref{thm:scaling-model-I} and \ref{thm:scaling-model-II} below). More precisely, the SDE \eqref{eq:SDE-schrodinger} and a time-changed version of \eqref{eq:SDE-sine} will be obtained as the scaling limit of the Pr\"ufer phases. On the other hand,
in the case of Model I at the critical value $\alpha=\frac{1}{2}$, we can identify the limiting shape of the eigenfunctions by looking at the scaling limit of the Pr\"ufer radii.
If we let $\Psi^{\lambda}_L=\Psi^{\lambda}_{L,+}\otimes \Psi^{\lambda}_{L,-}$ be a normalized eigenfunction corresponding to $\lambda \in \Upsilon_L$, then, by the scaling limits from the aforementioned theorems,
\begin{eqnarray*}
	\Big{\{} \left(\lambda, L |\Psi^{\lambda}_L(Lt)|^2 dt\right):\, \lambda\in\Upsilon_{\omega,L} \Big{\}}
	\Rightarrow
	\Big{\{} \left(\lambda, |\Psi^{\lambda}(t)|^2dt\right):\, \lambda \in \textbf{Sch}_{\sigma^2}\Big{\}},
\end{eqnarray*}
where 
\begin{eqnarray*}
	\left|
		\Psi^{\lambda}(t)
	\right|^2
	=
	\frac{e^{2r^{\lambda}(t)}}{\int^1_0 e^{2r^{\lambda}(s)}ds},
\end{eqnarray*}
and $r^{\lambda}$, which corresponds to the scaling limit of the Pr\"ufer radii, satisfies the system of SDEs \eqref{eq:SDEs-model-I} below.
It is possible to give a better characterization of the asymptotic shape of the eigenfunctions. The following theorem can be obtained from the scaling limits contained in Theorem \ref{thm:scaling-model-I}, combined with the general arguments of \cite{RV} which we will not repeat here. 
\begin{theorem}
	Consider Model I with $\alpha=\frac{1}{2}$. 
	Let $\mu$ be picked at random from the eigenvalues of $D_{\omega,L}$, let $\Psi^{\mu}_L=\Psi^{\mu}_{L,+}\otimes \Psi^{\mu}_{L,-}$ be the corresponding eigenfunction and write
	\begin{eqnarray*}
		|\Psi^{\mu}(j)|^2=|\Psi^{\mu}_+(j)|^2+|\Psi^{\mu}_-(j)|^2.
	\end{eqnarray*}
	Then,
	\begin{eqnarray*}
		\left(
			\mu, \, L|\Psi^{\mu}(Lt)|^2dt
		\right)
		\Rightarrow
		\left(
			E,
			\frac{S(\rho(E)^2(t-U))\, dt}{\int^1_0 S(\rho(E)^2(s-U))\, ds},
		\right),
	\end{eqnarray*}
	where 
	\begin{eqnarray*}
		S(t)
		=
		\exp\left( 
			\frac{\mathcal{B}_t}{\sqrt{2}}-\frac{|t|}{4}
		\right),
	\end{eqnarray*}
	$E$, $U$ and $\mathcal{B}$ are independent and such that
	\begin{itemize}
		\item $E$ is distributed according to $\rho(x)\, dx$,
		
		\vspace{1ex}
		
		\item $U$ is uniformly distributed on $[0,1]$,
		
		\vspace{1ex}
		
		\item $\mathcal{B}$ is a double-sided Brownian motion.
	\end{itemize}
\end{theorem}
In some sense, this is a weak localization type result as, even if one obtains a good decay of the eigenfunctions, the localization centers spread over the whole box. In contrast, we showed that Model II for $\alpha=\frac12$ is always delocalized, even in the pure point regime \cite{BMT01}.

\vspace{1ex}

\begin{remark}
\label{rk:Poisson}
Finally, we note that the case $\alpha<\frac12$, which is believed to yield Poisson statistics, is open, even in the case of the sub-critical discrete Anderson model. To the best of our knowledge, this result is currently known only in the case of a continuous Anderson model in a Brownian potential \cite{Kotani-Nakano17}, where stochastic calculus can be applied to reproduce the scheme of proof from \cite{AllezDumaz}, and in the case of unitary models presenting some type of rotational symmetry \cite{Stoiciu, Killip-Stoiciu}. 
Note that, in \cite{BMT01}, we proved dynamical localization for Model II and $\alpha<\frac12$, which is a key input for the classical proofs of Poisson statistics (see \cite[Chapter 17]{AW}). Such localization bounds should be possible to obtain for Model I as well.
The main issue to keep going forwards with this classical scheme consists in obtaining suitable Wegner and Minami's estimates, as the standard bounds diverge with the size of the system due to the fact that the coupling constant $L^{-\alpha}$ goes to zero. 
Nonetheless, Model I at $\alpha=\frac12$ should converge in the scaling limit (and not only its rescaled spectrum) to a continuum model , which suggests that it might be possible to obtain uniform Wegner and Minami's bounds, in the spirit of \cite[Proposition 6.4]{DumazLabbe}. 
Alternatively, we note that, in each box of size $L^{2\alpha}$, the system is critical. 
The full spectrum could then in principle be approximated by a triangular array of independent point processes restricted to vanishing rescaled intervals, yielding a possible path towards Poisson statistics. We defer these questions to a future work. 

The bottom of the spectrum is definitely not covered by the techniques of the present article. Note that, in this case, the density of states diverges. 
\end{remark}





\subsection{Scaling limit of the Pr\"ufer transform}

This section contains the scaling limits used in the proof of Theorem \ref{thm:main}.

In both models, the operator $\drcnr$ can be extended to the whole $\h$ with the boundary condition $\vp^-_0=0$. We denote this extension by $\drcnl$. A generalized solution to the eigenvalue problem for $\drcnl$ can be obtained in a recursive way by means of transfer matrices i.e. there exists a sequence of matrices $(\Tl_{L,n})_n$ such that, if
$\Phi_L(1)
=
\begin{pmatrix}
	1 \\ 0
\end{pmatrix},
$
and if $\Phi_L^\lambda(n+1):=\Tl_{L,n}\Phi_L^\lambda(n)$ for all $n$, then the vector $\Phi_L=(\Phi_L^\lambda(n))_n$ satisfies $\drcnl \Phi_L^\lambda  = \lambda \Phi_L$. The exact form of these matrices depends of course on the model.

To write down the transfer matrices, let $p_1(x)=m-E+x$ and $p_2(x)=m+E-x$, and define
\begin{eqnarray*}
	T(x_1,x_2)
	=
	\begin{pmatrix}
		p_1(x_1)p_2(x_2)+1 & p_2(x_2) \\
		p_1(x_1) & 1
	\end{pmatrix}.
\end{eqnarray*}
With this, it is not hard to see that the transfer matrices are given by $\Tl_{L,n}=\Tl(\tvon(n),\tvtn(n))$ where 
\begin{eqnarray*}
	\tvon(n)
	=
	\gamma
	\frac{\omega_{n,1}}{L^{\alpha}}-\frac{\lambda}{\rho L},
	\quad
	\tvtn(n)
	=
	\gamma
	\frac{\omega_{n+1,2}}{L^{\alpha}}-\frac{\lambda}{\rho L},
\end{eqnarray*}
for Model I and
\begin{eqnarray*}
	\tvon(n)
	=
	\gamma
	\frac{\omega_{n,1}}{n^{\alpha}}-\frac{\lambda}{\rho L},
	\quad
	\tvtn(n)
	=
	\gamma
	\frac{\omega_{n+1,2}}{n^{\alpha}}-\frac{\lambda}{\rho L},
\end{eqnarray*}
for Model II.


With this formalism, we can see that $\Phil_L=\phi^{\lambda,+}_L\otimes \phi^{\lambda,-}_L$ defined above is an eigenvector of $\drcnr$ with corresponding eigenvalue $\lambda$ if and only if $\phi^{\lambda,+}_L(L)=0$. We now introduce the Pr\"ufer transform to write this last condition in a convenient form. 

Let $E\in(m,\sqrt{m^2+4})$ and let $k\in(-\pi,-\pi/2)$ be defined through the relation
\begin{eqnarray}
\label{eq:Ek}
	\cos k = - \frac{1}{2} \sqrt{E^2-m^2}.
\end{eqnarray}
Let 
\begin{eqnarray}\label{eq:change-of-basis}
	\ppru_n
	=
	(-1)^{n-1}
	\begin{pmatrix}
		-\sqrt{p_2} \cos((2n-1)k) & -\sqrt{p_2} \sin((2n-1)k)\\
		\sqrt{-p_1} \cos((2n-2)k) & \sqrt{-p_1} \sin((2n-2)k)
	\end{pmatrix}.
\end{eqnarray}
If $\Phil_L$ is defined as above, we define a new sequence $(\Psil_L(n))_n$ through the relation $\Phil_L(n) = \ppru_n \Psil_L(n)$. Writing $\Psil_L = \psi^{\lambda,+}_L\otimes \psi^{\lambda,-}_L$, we define the Pr\"ufer coordinates as
\begin{equation*}
	\zl_L(n)=\Rl_L(n) e^{i\tl_L(n)} := \psi^{\lambda,+}_L(n) + i \psi^{\lambda,-}_L(n).
\end{equation*}
Here, the phase $\tl_L$ is defined in such a way that $\tl_L(0)=0$ and $|\tl_L(j)-\tl_L(j-1)| < 2\pi $. Thanks to the decay of the environment, it can be seen that this procedure is well-defined (see \cite[Sec 3]{BMT01}).
Recall that $\lambda \in \spn$ if and only if
\begin{eqnarray*}
	\Psi_L^{\lambda}(L)
	=
	c
	\begin{pmatrix}
		0 \\ 1
	\end{pmatrix},
\end{eqnarray*}
for some constant $c\neq 0$.
In terms of Pr\"ufer coordinates, this  can be written as
\begin{eqnarray*}
	\Rl_L(L) e^{i\tl_L(L)}
	=
	c
	\ppru_L^{-1}
	\begin{pmatrix}
		0 \\ 1
	\end{pmatrix}
\end{eqnarray*}
Explicitly,
\begin{eqnarray*}
	e^{i\tl_L(L)} = \pm i e^{i(2L-1)k},
\end{eqnarray*}
or, equivalently,
\begin{eqnarray*}
	\tl_L(L) - [(2L-1)k] - \tfrac{\pi}{2} \in \pi \Z.
\end{eqnarray*}
Hence, we can rewrite our point process as
\begin{eqnarray*}
	\Upsilon_L
	=
	\left\{
		\lambda:\, \tl_L(L) - [(2L-1)k] - \tfrac{\pi}{2} \in \pi \Z
	\right\}.
\end{eqnarray*}
It will be convenient to rewrite this conditions as
\begin{eqnarray}\label{eq:spectrum prufer}
	\Upsilon_L
	=
	\left\{
		\lambda:\, \vtl_L(L) -2 [(2L-1)k] -\pi  \in 2 \pi \Z
	\right\},
\end{eqnarray}
with $\vtl_L = 2\tl_L$.
Theorem 2.3 will be a corollary of the following scaling limit. As above, we only consider the positive component of the spectrum as the other case is similar and omit the energy level $E=\sqrt{m^2+2}$ (see \cite{KVV}).
\begin{theorem}\label{thm:scaling-model-I}
	Consider Model I with $\alpha=\frac12$ and let $E\in(m,\sqrt{m^2+4})$, $E\neq \sqrt{m^2+2}$.
	Let $\vtl_L(t)=2\theta^{\lambda}_L(Lt)$ and $r^{\lambda}_L(t)=\log R_L^{\lambda}(Lt)$. Then, the sequence
	$\{\vartheta^{\lambda}_L,r^{\lambda}_L:\, t\in[0,T],\, \lambda\in\R\}$ converges weakly to the unique solution of the SDE
	\begin{eqnarray}
		d\vartheta^{\lambda}
		&=&
		\lambda dt
		+
		\sigma
		\text{Re}\left\{
			e^{-i\vartheta^{\lambda}}dB
		\right\}
		+
		\sigma dW
		\\
		dr^{\lambda}
		&=&
		\frac{\sigma^2}{8}dt
		+
		\frac{\sigma}{2}
		\text{Im}\left\{
			e^{-i\vartheta^{\lambda}}dB
		\right\},
		\label{eq:SDEs-model-I}
	\end{eqnarray}
	in the topology of uniform convergence over $[0,T]$ and in the sense of finite dimensional distributions in $\lambda$, where $B$ and $W$ are independent standard complex Brownian motion and standard Brownian motion respectively, and $\sigma^2 = \frac{p_1(E)^2+p_2(E)^2}{\sin^2(2k)}$.

	 Furthermore, the sequence of processes $(\lambda \mapsto \vartheta_L^{\lambda}(1))_L$ converges weakly to $\lambda \mapsto \vartheta^{\lambda}(1)$ in the topology of uniform convergence over compact sets. 
	 
	 \vspace{1ex}
	
	If $\alpha>\frac12$ or $\gamma=0$, the same holds with respect to the limiting process $\vtl(t)=\vtl(0) + \lambda t$.
\end{theorem}

We now state our result for Model II. In this case, it is convenient to reverse the envelope environment i.e. to consider
\begin{eqnarray*}
	\tvon(n)
	=
	\gamma
	\frac{\omega_{n,1}}{(L-n)^{\alpha}}-\frac{\lambda}{\rho L},
	\quad
	\tvtn(n)
	=
	\gamma
	\frac{\omega_{n+1,2}}{(L-n)^{\alpha}}-\frac{\lambda}{\rho L}.
\end{eqnarray*}
In terms of the law of the objects associated to the random operator $\mathcal{D}_L$, this change is harmless but produces a limiting SDE with a singularity at $t=1$ instead of $t=0$.
\begin{theorem}\label{thm:scaling-model-II}
	Consider Model II with $\alpha=\frac12$ and let $E\in(m,\sqrt{m^2+4})$, $E\neq \sqrt{m^2+2}$.
	Let $\vtl_L(t)=2\theta^{\lambda}_L(Lt)$ and $r^{\lambda}_L(t)=\log R_L^{\lambda}(Lt)$. Then, the sequence
	$\{\vartheta^{\lambda}_L,r^{\lambda}_L:\, t\in[0,1),\, \lambda\in\R\}$ converges weakly to the unique solution of the SDE
	\begin{eqnarray}
	\label{eq:SDE-sine-2}
		d\vartheta^{\lambda}
		&=&
		\lambda dt
		+
		\frac{\sigma}{\sqrt{1-t}}
		\text{Re}\left\{
			e^{i\vartheta^{\lambda}}dB
		\right\}
		+
		\sigma dW,
	\end{eqnarray}
	in the topology of uniform convergence over compact subset of $[0,1)$ and in the sense of finite dimensional distributions in $\lambda$, where $B$ and $W$ are independent standard complex Brownian motion and standard Brownian motion respectively, and $\sigma^2 = \frac{p_1(E)^2+p_2(E)^2}{\sin^2(2k)}$.
	
%
%
%
%
	If $\alpha>\frac12$ or $\gamma=0$, the same holds with respect to the limiting process $\vtl(t)=\vtl(0) + \lambda t$.
\end{theorem}
To obtain \eqref{eq:SDE-sine} from \eqref{eq:SDE-sine-2}, it is enough to perform the time change $t\mapsto 1-e^{-\beta t /4}$, with $\beta = \frac{2}{\sigma^2}$. See \cite{KVV}. 

In the rest of this work, we will give full proofs in the case of Model I while, for Model II, we will only sketch the arguments and refer the reader to \cite{KVV} whenever further details are needed.


\section{The integrated density of states}

 The integrated density of states for the whole free operator $\drcf$ corresponds to
\begin{eqnarray*}
	\dos_0(\varphi)
	&=&
	\lim_{L\to\infty}
	\frac{1}{2L}
	\sum_{\substack{1\leq j \leq L \\ \sigma=\pm}} \bra \delta^{\sigma}_j, \varphi(\drcf) \delta^{\sigma}_j \ket.
\end{eqnarray*}
The usual proof of the existence of this limit for ergodic self-adjoint operators uses restrictions to boxes with suitable boundary conditions and monotonicity/subadditivity arguments, which become delicate for our perturbed operators. 
Furthermore, the limit can be recovered as
\begin{eqnarray}\label{eq:dos-from-boxes}
	\dos_0(\varphi)
	=
	\lim_{L\to\infty}
	\frac{1}{2L}
	\sum_{\substack{1\leq j \leq L \\ \sigma=\pm}} \bra \delta^{\sigma}_j, \varphi(\drcnf) \delta^{\sigma}_j \ket,
\end{eqnarray}
independently of the (reasonable) boundary conditions imposed on $\drcnf$. In the following, we will consider $\varphi^-(0)=\varphi^+(L)=0$.
It turns out that this procedure will allow us to compute the local density of states i.e. the density of the measure $\dos_0$.
In many cases, this can be achieved by an application of Sz\"ego's theorem. We follow a different paths that builds on our knowledge of the Pr\"ufer phases.

We will begin with a quantitative estimate on the limit \eqref{eq:dos-from-boxes}.
Let $R_0(z)=(D_0-z)^{-1}$ and $R_{0,L}(z)=(D_L-z)^{-1}$ be the resolvents of the free operator $D$ and the restricted operator  $D_L$ respectively.  We define the corresponding Green's functions by
$$ G_0(z;j,\sigma;j',\sigma')=\bra\delta_j^\sigma,R_0(z)\delta_{j'}^{\sigma'}\ket,\quad 
G_{0,L}(z;j,\sigma;j',\sigma')=\bra\delta_j^\sigma,R_{0,L}(z)\delta_{j'}^{\sigma'}\ket.$$
We start with a preliminary estimate.

\begin{lemma}\label{thm:comparison-resolvents}
	There is a constant $C>0$ such that
	\begin{eqnarray*}
		\sum_{\substack{1\leq j \leq L \\ \sigma=\pm}} 
		\left|
			G_0(z;j,\sigma;j,\sigma)-G_{0,L}(z;j,\sigma;j,\sigma)
		\right|
		\leq
		C |z|^{-2},
	\end{eqnarray*}
	for all $z\notin \R$.
\end{lemma}
\begin{proof}
	Note that
	\begin{eqnarray*}
	D_0
	=
	D_{0,L}
	- \deltap_L \bra \deltam_{L+1}| 
	+
	P_L^{\perp} C P_L^{\perp},
	\end{eqnarray*}
	for some bounded operator $C$, where $P_L^{\perp}$ denotes the projection on $\bra \{  \delta^{\pm}_j:\, |j|\geq L \}\ket$. By the resolvent formula,
	\begin{eqnarray}
		&&
		\sum_{\substack{1\leq j \leq L \\ \sigma=\pm}} 
		\left(
			G_0(z;j,\sigma;j,\sigma)-G_{0,L}(z;j,\sigma;j,\sigma)
		\right)
		\\
		\label{eq:res-formula}
		&=&
		-\sum_{\substack{1\leq j \leq L \\ \sigma=\pm}} 
		\Big{(}
			G_{0,L}(z;j,\sigma;L,+)G_0(z;L+1,-;j,\sigma)
		\Big{)}.
	\end{eqnarray}
	Finally, by Cauchy-Schwarz inequality,
	\begin{eqnarray*}
		&&
		\sum_{\substack{1\leq j \leq L \\ \sigma=\pm}} 
			\left| G_{0,L}(z;j,\sigma;L,+)G_0(z;L+1,-,\sigma;j,\sigma) \right|
		\\
		&&
		\phantom{blablablabla}
		\leq
		\| R_{0,L}(z) \deltap_L \| \| R_0(z)\deltam_{L+1} \|
		\leq C
		|z|^{-2}.
	\end{eqnarray*}
\end{proof}
Let $\dosnf$ denote the integrated  density of states of $\drcnf$.
\begin{lemma}\label{thm:limit-DOS-boxes}
	For each smooth $\varphi$, there is a constant $C=C(\varphi)$ such that
	\begin{eqnarray*}
		\left|
			\dos_0(\varphi)- \dos_{0,L}(\varphi)
		\right|
		\leq
		C L^{-1}.
	\end{eqnarray*}
\end{lemma}
\begin{proof}
	First, note that
\begin{align*}
\left|\dos_0(\varphi)- \dos_{0,L}(\varphi)\right|
&\le
\frac{1}{2L}
 \sum_{\substack{1\leq j \leq L \\ \sigma=\pm}} |\bra \delta^{\sigma}_j, (\varphi(\drcf)-\varphi(D_L))\delta^{\sigma}_j \ket|.
\end{align*}
The proof then follows from the Helffer-Sjostrand formula \cite{Davies}: if $\tilde{\varphi}$ is a quasi-analytic extension of $\varphi$, then
\begin{eqnarray*}
	\varphi(\drcf)-\varphi(D_L)
	&=&
	-\frac{1}{\pi}
	\int_{\mathbb{C}} \frac{\partial}{\partial \overline{z}} \tilde{\varphi}(z)
	\left(
		R_0(z)-R_{0,L}(z)
	\right)
	\lambda(dz),
\end{eqnarray*}
where $\lambda$ denotes the Lebesgue measure on $\mathbb{C}$. Using the estimate in Lemma \ref{thm:comparison-resolvents}, we obtain that
\begin{eqnarray*}
	\frac{1}{2L}
	\sum_{\substack{1\leq j \leq L \\ \sigma=\pm}} 
	|\bra \delta^{\sigma}_j, (\varphi(\drcf)-\varphi(D_L))\delta^{\sigma}_j \ket|
	&\leq&
	\frac{C}{L}
	\frac{1}{\pi}
	\int_{\mathbb{C}} \left| \frac{\partial}{\partial \overline{z}} \tilde{\varphi}(z)\right|
		|z|^{-2}
	\lambda(dz).
\end{eqnarray*}
The integral is finite thanks to the property of quasi-analytic extensions. This finishes the proof.
\end{proof}
\begin{proposition}
	The integrated density of states $\dos_0$ admits the explicit density $(2\pi)^{-1}\rho(E)$ for all $m < |E| < \sqrt{m^2 + 4}$. 
\end{proposition}
\begin{proof}
	We use the convergence of the Prüfer transform from Theorem \ref{thm:scaling-model-I} with $\gamma=0$ (which proof is independent of this proposition).
	Let $L\geq 1$ and let $\{\tl_L(j):\, j\geq 1\}$ denote the Pr\"ufer phase with base energy $E_0$ for the free system. Then, according to \eqref{eq:spectrum prufer}, we have
	\begin{eqnarray*}
		L\dos_{0,L}([E_0,E_0+\tfrac{2\pi}{\rho(E_0)L}))
		&=&
		\# \left\{
			\lambda\in[0,2\pi):\,
			\vtl_L(L)-2[(2L-1)k]-\pi \in 2\pi\Z
		\right\}
		\\
		&=&
		\left\lfloor
			\frac{\vartheta^{2\pi}_L(L)-2[(2L-1)k]-\pi}{2\pi}
		\right\rfloor
		\\
		&& \quad
		- 
		\left\lceil
			\frac{\vartheta^0_L(L)-2[(2L-1)k]-\pi}{2\pi}
		\right\rceil
		+
		1.
	\end{eqnarray*}
	Now, as $\tl_L(L) \to \theta_0 + \lambda$ uniformly in compacts, for some constant $\theta_0$, it must hold that
	\begin{eqnarray}
		\lim_{L\to\infty}
		L\dos_{0,L}([E_0,E_0+\tfrac{2\pi}{\rho(E_0)L}])
		=
		1.
	\end{eqnarray}
	Furthermore, as the left-hand side is integer-valued, there exists $L_0\geq 1$ such that
	\begin{eqnarray*}
		L\dos_{0,L}([E_0,E_0+\tfrac{2\pi}{\rho(E_0)L}])
		=
		1,
	\end{eqnarray*}
	for all $L\geq L_0$. Note that, thanks to the uniform convergence of the Pr\"ufer phases, such an $L_0=L_0(E_0)$ can be chosen uniformly bounded for $E_0$ running over compact sets contained in the interior of the spectrum of $D_0$.
	Hence, for all $\varepsilon>0$, there exists $L_{\varepsilon}\geq 1$ such that
	\begin{eqnarray*}
		L\dos_{0,L}([E,E+\tfrac{2\pi}{\rho(E)L}])
		=
		1,
	\end{eqnarray*}
	for all $E\in[E_0,E_0+\varepsilon]$ and all $L\geq L_{\varepsilon}$.
	
	Let $\rho^L_0=\rho(E_0)$, $E^L_1 = E^L_0 + \frac{2\pi}{\rho^L_0 L}$ and, in general, $\rho^L_j=\rho(E^L_j)$ and $E^L_{j+1}=E^L_j + \frac{2\pi}{\rho^L_j L}$. For $\varepsilon>0$, let
	\begin{eqnarray*}
		M_L(\varepsilon)
		=
		\min\left\{
			m:\, E^L_m \geq E_0+ \varepsilon
		\right\}.
	\end{eqnarray*}
	Then,
	\begin{eqnarray*}
		\dos_0([E_0,E_0+\varepsilon])
		&=&
		\lim_{L\to\infty} 
		\dos_{0,L}([E_0,E_0+\varepsilon))
		\\
		&=&
		\lim_{L\to\infty} 
		\dos_{0,L}([E_0,E_{M_L(\varepsilon)}])
		\\
		&=&
		\lim_{L\to\infty} 
		\sum^{M_L(\varepsilon)}_{j=1}
		\dos_{0,L}([E^L_{j-1},E^L_{j-1}+\frac{2\pi}{\rho^L_{j-1}L}))
		\\
		&=&
		\lim_{L\to\infty} 
		\frac{M_{L}(\varepsilon)}{L}.
	\end{eqnarray*}		
	Now,
	\begin{eqnarray*}
		E^L_{M_L(\varepsilon)}
		&=&
		E_0 + \sum^{M_L(\varepsilon)}_{j=0} \frac{2\pi}{\rho^L_j L}
		\geq 
		E_0 + \varepsilon
		>
		E_0 + \sum^{M_L(\varepsilon)-1}_{j=0} \frac{2\pi}{\rho^L_j L}.
	\end{eqnarray*}
	By continuity of $\rho$, there exists $c>0$ such that $|\rho_j^L-\rho_0|<c\varepsilon$ for all $j\leq M_L(\varepsilon)$. Hence,
	\begin{eqnarray*}
		\frac{2\pi M_L(\varepsilon)}{(\rho_0-c\varepsilon) L}
		=
		\sum^{M_L(\varepsilon)}_{j=0} \frac{2\pi}{(\rho_0-c\varepsilon) L}
		\geq 
		\varepsilon
		\geq
		\sum^{M_L(\varepsilon)-1}_{j=0} \frac{2\pi}{(\rho_0+c\varepsilon) L}
		=
		\frac{2\pi(M_L(\varepsilon)-1)}{(\rho_0+c\varepsilon) L}.
	\end{eqnarray*}
	It then follows that
	\begin{eqnarray*}
		\frac{2\pi }{\rho_0-c\varepsilon}
		\lim_{L\to\infty}
		\frac{M_L(\varepsilon)}{ L}
		\geq
		\varepsilon
		\geq
		\frac{2\pi}{\rho_0+c\varepsilon}
		\lim_{L\to\infty}
		\frac{M_L(\varepsilon)}{L},
	\end{eqnarray*}
	and
	\begin{eqnarray*}
		\lim_{\varepsilon\to0}
		\varepsilon^{-1}
		\lim_{L\to\infty}
		\frac{M_L(\varepsilon)}{L}
		=
		\frac{\rho_0}{2\pi}.
	\end{eqnarray*}
	This finishes the proof.

\end{proof}

Next, we compare the free and perturbed operators. The coincidence of the density of states for the free and perturbed operators was first observed by Dolai and Krishna for $d\geq 3$ \cite{DolaiKrishna}.


\begin{lemma}
	\begin{eqnarray*}
		\lim_{L\to\infty}\dos_{\omega,L}=\dos_{0}.
	\end{eqnarray*}
\end{lemma}
\begin{proof}
	We start noticing that
	\begin{eqnarray*}
		D_{\omega,L}
		=
		D_{0,L}
		+
		\sum_{|j|\leq L}
		\left(
			\vo(k) \deltap_k \bra \deltap_k |+ \vt(k) \deltam_k \bra \deltam_k |
		\right).
	\end{eqnarray*}
	By the resolvent formula,
	\begin{eqnarray*}
		&&
		\sum_{|j|\leq L, \sigma=\pm}
		\left|
			G_{\omega,L}(z;j,\sigma;j,\sigma)
			-
			G_{0,L}(z;j,\sigma;j,\sigma)
		\right|
		\\
		&\leq&
		\sum_{|j|\leq L, \sigma=\pm}
		\Big{|}
		\sum_{|k|\leq L}
		\Big{(}
			G_{\omega,L}(z;j,\sigma;k,+)
			G_{0,L}(z;k,+;j,\sigma)\vo(k)
			\\
			&&
			\phantom{blablabla}
			+
			G_{\omega,L}(z;j,\sigma;k,+)
			G_{0,L}(z;k,+;j,\sigma)\vt(k)
		\Big{)}
		\Big{|}.
	\end{eqnarray*}
	Now, by Cauchy-Schwarz inequality,
	\begin{eqnarray*}
		&&
		\sum_{|j|\leq L, \sigma=\pm}
		\Big{|}
		\sum_{|k|\leq L}
		\Big{(}
			G_{\omega,L}(z;j,\sigma;k,+)
			G_{0,L}(z;k,+;j,\sigma)\vo(k)
		\Big{)}
		\Big{|}
		\\
		&\leq&
		\sum_{|k|\leq L}
		\left| \vo(k) \right|
		\sum_{|j|\leq L, \sigma=\pm}
		\Big{|}
			G_{\omega,L}(z;j,\sigma;k,+)
			G_{0,L}(z;k,+;j,\sigma)
		\Big{|}
		\\
		&\leq&
		\sum_{|k|\leq L}
		\left| \vo(k) \right|
		\| R_{\omega,L}(z) \deltap_k\| \| R_{0,L}(z) \deltap_k\|
		\\
		&\leq&
		C |z|^{-2} \sum_{|k|\leq L}
		\left| \vo(k) \right|.
	\end{eqnarray*}
	The last sum is $o(L)$. The sum involving the random variables $\vt(k)$ can be treated in the same way. The result then follows as in the proof of Lemma \ref{thm:limit-DOS-boxes}.
\end{proof}

\section{Scaling limit for Model I}

We fix a base energy $E_0\in(m,\sqrt{m^2+4})$ corresponding to $k\in(-\pi,-\frac{\pi}{2})$, $k\neq-\frac{3\pi}{4}$ through \eqref{eq:Ek} and we write $\rho=\rho(E_0)$.
We denote the Pr\"ufer transform for $\mathcal{D}_L$ corresponding to the energy level $E_0 + \frac{\lambda}{\rho L}$ by $\zeta^{\lambda}_L(j)=R^{\lambda}_L(j)e^{i\theta^{\lambda}_L(j)}$, denote $\eta_j=(2j-1)k$ and $\bar{\theta}^{\lambda}_L(j)=\theta^{\lambda}_L(j)-\eta_j$. Let
\begin{eqnarray*}
	\sigma_1=\frac{-p_1}{\sin(2k)},
	\quad
	\sigma_2=\frac{p_2}{\sin(2k)},
	\quad
	\sigma^2=\sigma_1^2+\sigma_2^2.
\end{eqnarray*}
The recursion for $\zeta^{\lambda}_L$ can be found in \cite[Equation (3.17)]{BMT01} and is given  by
\begin{eqnarray*}
	\zeta^{\lambda}_L(j+1)
	&=&
	\Big{(}
	1
	- i\sigma_2
	\cos \bar\theta^{\lambda}_L(j) \, 
	e^{-i\bar\theta^{\lambda}_L(j)}
	  \left(
		 \frac{\omega_{1,j}}{L^{\alpha}} - \frac{\lambda}{\rho L}
	\right)
	\\
	&&
	\phantom{blablabla}
	- i\sigma_1
	\cos( \bar\theta^{\lambda}_L(j)-k)
	\, e^{-i (\bar\theta^{\lambda}_L(j)-k)}
	\left(
		 \frac{\omega_{2,j+1}}{L^{\alpha}} - \frac{\lambda}{\rho L}
	\right)
	\\
	\nonumber
	&&
	\phantom{blablabla}
	-i \sqrt{\sigma_1\sigma_2}
	\cos \bar\theta^{\lambda}_L(j) \,
	e^{-i (\bar \theta^{\lambda}_L(j)-k)}
	\left(
		 \frac{\omega_{1,j}}{L^{\alpha}} - \frac{\lambda}{\rho L}
	\right)
	\left(
		 \frac{\omega_{2,j+1}}{L^{\alpha}} - \frac{\lambda}{\rho L}
	\right)
	\Big{)}\,
	\zeta^{\lambda}_L(j)
	\\
	&=:&
	\Big{(}
		1 +\Gamma_j(\lambda,\theta_L(j))
	\Big{)} 
	\, \zeta^{\lambda}_L(j),
\end{eqnarray*}
with $R^{\lambda}_L(0)=1$ and $\theta^{\lambda}_L(j)=0$.
Note that, for $L$ large enough, the quantity in the parenthesis above will be contained in a small neighbourhood of $1$ for all $j=0,\cdots,L$, so that its logarithm will be well defined and $\zeta_L(j)$ will not vanish. We let $r^{\lambda}_L(j)=\log R^{\lambda}_L(j)$, so that $r^{\lambda}_L(0)=0$.

\smallskip

\begin{proof}[Proof of Proposition \ref{thm:scaling-model-I}]
Let $\lambda_1,\cdots,\lambda_m\in\R$, write $\theta_{L,l}(j)=\theta^{\lambda_l}_L(j)$ for $l=1,\cdots,m$, $r_{L,l}(j)=r^{\lambda_l}_L(j)$ and consider the $2m+2$-dimensional recursion
\begin{eqnarray*}
	\theta_{L,l}(j+1)
	&=&
	\theta_{L,l}(j)
	+
	\text{Im}
	\left\{
	\log\Big{(}
		1 + \Gamma_j(\lambda_l,\theta_{L,l}(j))
	\Big{)}
	\right\},
	\\
	r_{L,l}(j+1)
	&=&
	r_{L,l}(j)
	+
	\text{Re}
	\left\{
	\log\Big{(}
		1 + \Gamma_j(\lambda_l,\theta_{L,l}(j))
	\Big{)}
	\right\},
	\quad
	l=1,\cdots,m,
	\\
	B_{L,1}(j+1)
	&=&
	B_{L,1}(j) - \frac{\omega_{1,j}}{L^{\alpha}},
	\\
	B_{L,2}(j+1)
	&=&
	B_{L,2}(j) - \frac{\omega_{2,j+1}}{L^{\alpha}}.
\end{eqnarray*}
	In the context of Theorem \ref{thm:general-convergence-finite-dimensional},
	we denote 
	\begin{eqnarray*}
	X^L
	=
	\Big{(}X^L_l\Big{)}_{l=1}^{2m+2}
	=
	\Big{(}\theta_{L,1},\cdots, \theta_{L,m},
	r_{L,1},\cdots,r_{L,m}
	B_{L,1},B_{L,2}\Big{)},
	\end{eqnarray*}
	and we let
	$$
	Y^L_j(x)
	=
	\Big{(}Y^L_{j,l}\Big{)}_{l=1}^{2m+2}
	=
	\Big{(}Y^L_{j,1}(x), \cdots, Y^L_{j,2m+2}(x)\Big{)}$$
	to be the increment $X^L_{j+1}-X^L_j$ conditioned on $X^L_j=x\in\R^{2m+1}$:
	\begin{eqnarray*}
		Y^L_{j,l}(x)
		&=&
		\text{Im}
		\left\{
		\log\Big{(}
			1 + \Gamma_j(\lambda_l,x_l)
		\Big{)}
		\right\},
		\\
		Y^L_{j,m+l}(x)
		&=&
		\text{Re}
		\left\{
		\log\Big{(}
			1 + \Gamma_j(\lambda_l,x_l)
		\Big{)}
		\right\},
	\end{eqnarray*}
	for $l=1,\cdots,m$ and
	\begin{eqnarray*}
		Y^L_{j,2m+1}(x)
		=
		- \frac{\omega_{1,j}}{L^{\alpha}},
		\quad
		\quad
		Y^L_{j,2m+2}(x)
		=
		- \frac{\omega_{2,j+1}}{L^{\alpha}}.
	\end{eqnarray*}
	We define
	\begin{eqnarray*}
		b^L_l(t,x)
		=
		L \esp[Y^L_{Lt,l}(x)],
		\quad
		a^L_{l,l'}(t,x)
		=
		L\esp[Y^L_{Lt,l}(x)_lY^L_{Lt,l'}(x)],
	\end{eqnarray*}
	where, with a slight abuse of notation, each time a continuous parameter $s$ is given where an integer one is expected, it is understood that $s$ has to be replaced by $\lfloor s \rfloor$. Let us write
	$$ b^L(t,x)=\Big{(}b^L_l(t,x)\Big{)}_l,\quad a^L(t,x)=\Big{(}a^L_{l,l'}(t,x)\Big{)}_{l,l'}.$$
	Let us start computing the drift terms i.e. the limits
	\begin{eqnarray}
		\label{eq:general-drift-limit}
		\lim_{L\to\infty}
		\int^t_0 b^L_l(s,x)\, ds
		=
		\lim_{n\to\infty}
		\frac{1}{L} \sum^{Lt}_{j=1} b^L_l(\tfrac{j}{L},x)ds
		=
		\lim_{n\to\infty}
		\sum^{Lt}_{j=1}
		\esp\left[
			Y^L_{j,l}(x)
		\right].
	\end{eqnarray}
	 As the random variables $(\omega_{i,j})_j$, $i=1,2$ are centered, this limit is $0$ for $l=2m+1$ and $2m+2$.
	For $l=1,\cdots,2m$, we use the approximation $\log(1+z) = z - \frac{z^2}{2}+O(|z|^3)$ to obtain
	\begin{eqnarray}
		\label{eq:log-approximation}
		\log\Big{(} 1 + \Gamma_j(\lambda_l,x_l) \Big{)}
		&=&
		-i\sigma_2
		\cos(x_l-\eta_j) \, 
		e^{-i(x_l-\eta_j)}
		 \left(
		 	\frac{\omega_{1,j}}{L^{\alpha}} - \frac{\lambda_l}{\rho L}
		 \right)
		\\
		\nonumber
		&&
		-
		i\sigma_1
		\cos(x_l-\eta_j+k)
		\, e^{-i (x_l-\eta_j+k)}
		\left(
		 	\frac{\omega_{2,j+1}}{L^{\alpha}} - \frac{\lambda_l}{\rho L}
		 \right)
		\\
		\nonumber
		&&
		-
		\frac{\sigma_2^2}{2}
		\cos^2(x_l-\eta_j) \, 
		e^{-i2(x_l-\eta_j)}
		\frac{\omega_{1,j}^2}{L^{2\alpha}}
		\\
		\nonumber
		&&
		-
		\frac{\sigma_1^2}{2}
		\cos^2(x_l-\eta_j+k)
		\, e^{-2i (x_l-\eta_j+k)}
		 	\frac{\omega_{2,j+1}^2}{L^{2\alpha}} 
		 +
		 A_j \omega_{1,j}\omega_{2,j}
		 +
		 O\left(
		 	\frac{1}{L^{3\alpha}}
		 \right),
	\end{eqnarray}
	for some sequence $(A_j)_j$.
	Note that the error terms above will have a vanishing contribution to the limits considered in \eqref{eq:general-drift-limit}. Likewise, the cross-terms $\omega_{1,j}\omega_{2,j}$ have zero expected value. Let us consider each term separately. First,
	\begin{eqnarray*}
		&&
		\lim_{L\to\infty}
		\sum_{j=1}^{Lt}
		\esp\left[
			i\cos(x_l-\eta_j) \, 
			e^{-i(x_l-\eta_j)}
		 	\left(
		 		\frac{\omega_{1,j}}{L^{\alpha}} - \frac{\lambda_l}{\rho L}
			 \right)
		\right]
		\\
		&& \quad
		=
		-
		\lim_{L\to\infty}
		\frac{\lambda_l}{\rho L}
		\sum_{j=1}^{Lt}
		\Big{(}
		i\cos^2(x_l-\eta_j)
		+
		\sin(x_l-\eta_j)
		\cos(x_l-\eta_j)
		\Big{)}
		=
		-
		\frac{i\lambda_l t}{2 \rho},
	\end{eqnarray*}
	for any $k\in(-\pi,-\frac{\pi}{2})$, where we used Lemma \ref{thm:trigonometric-limits-1}. The second term in \eqref{eq:log-approximation} can be handled in the same way. Next,
	\begin{eqnarray*}
		&&
		\lim_{L\to\infty}
		\sum^{Lt}_{j=1}
		\esp\left[
			\cos^2(x_l-\eta_j) \, 
			e^{2i(x_l-\eta_j)}
			\frac{\omega_{1,j}^2}{L^{2\alpha}}
		\right]
		\\
		&& \quad
		=
		\lim_{L\to\infty}
		\frac{1}{L^{2\alpha}}
		\sum^{Lt}_{j=1}
		\Big{(}
			\cos^2(x_l-\eta_j)
			\cos(2(x_l-\eta_j))
			+
			i
			\cos^2(x_l-\eta_j) \, 
			\sin(2(x_l-\eta_j))
		\Big{)}
		\\ 
		&& \quad
		=
		\left\{
			\begin{array}{ll}
				0 & \alpha>\frac{1}{2},
				\\
				&
				\\
				-\frac{t}{4} & \alpha=\frac{1}{2},
			\end{array}
		\right.
	\end{eqnarray*}
	for $k\neq -\frac{3\pi}{4}$, as can be seen from Lemma \ref{thm:trigonometric-limits-1} and the identities
	\begin{eqnarray*}
		\cos^2(\alpha)\cos(2\alpha)
		&=&
		2\cos^4 (\alpha)-\cos^2(\alpha),
		\\
		\cos^2(\alpha)\, \sin(2\alpha)
		&=&
		\frac{1}{2}\sin(4\alpha)
		+
		\frac{1}{4}\sin(4\alpha).
	\end{eqnarray*}
	The last term in \eqref{eq:log-approximation} can be handled along the same lines.
	Hence,
	\begin{eqnarray*}
		\lim_{L\to\infty}
		\sum^{Lt}_{j=1}
		\esp\left[
			\log\Big{(} 1 + \Gamma_j(\lambda_l,x_l) \Big{)}
		\right]
		=
		\frac{i\sigma_2 \lambda_l t}{2\rho}
		+\frac{i\sigma_1 \lambda_l t}{2\rho}
		+\frac{\sigma_2^2 t}{8}
		+\frac{\sigma_1^2 t}{8}
		=
		\frac{i \lambda_l t}{2}
		+
		\frac{\sigma^2 t}{8},
	\end{eqnarray*}
	for $\alpha=\frac{1}{2}$, while the second summand is replaced by $0$ if $\alpha>\frac{1}{2}$. As a conclusion,
	\begin{eqnarray*}
		\lim_{L\to\infty} \int^t_0 b^L(s,x)\, ds
		=
		\left\{
			\begin{array}{ll}
				\left(
					\frac{\lambda_1 t}{2},\cdots, \frac{\lambda_m t}{2},
					\frac{\sigma^2}{8},\cdots, \frac{\sigma^2}{8},
					0,0
				\right)
				&
				\alpha=\frac{1}{2},
				\\
				&
				\\
				\left(
					\frac{\lambda_1 t}{2},\cdots, \frac{\lambda_m t}{2},
					0,\cdots, 0,
					0,0
				\right)
				&
				\alpha>\frac{1}{2}.
			\end{array}
		\right.
	\end{eqnarray*}
	
	We now turn to the correlations i.e. to the limits
	\begin{eqnarray}
		\label{eq:general-diffusion}
		\lim_{L\to\infty}
		\int^t_0 a^L_{l,l'}(s,x)\, ds
		=
		\lim_{L\to\infty}
		\frac{1}{L}
		\sum^{Lt}_{j=1} a^L(\tfrac{j}{L},x)
		=
		\lim_{L\to\infty}
		\sum^{Lt}_{j=1}
		\esp\left[
			Y^L_{j,l}(x)\, Y^L_{j,l'}(x)
		\right].
	\end{eqnarray}
	First, we have that
	\begin{eqnarray*}
		&&
		\text{Im}\left\{\log\Big{(} 1 + \Gamma_j(\lambda_l,x_l) \Big{)}\right\}
		\text{Im}\left\{\log\Big{(} 1 + \Gamma_j(\lambda_{l'},x_{l'}) \Big{)}\right\}
		\\
		&&
		=
		\left(
		\sigma_2
		\cos^2(x_l-\eta_j) \, 
		 	\frac{\omega_{1,j}}{L^{\alpha}} 
		+
		\sigma_1
		\cos^2(x_l-\eta_j+k)
		 	\frac{\omega_{2,j+1}}{L^{\alpha}}
		 +
		 O\left(
		 	\frac{1}{L^{2\alpha}}
		 \right)
		 \right)
		 \\
		 &&
		 \quad
		 \times
		 \left(
		\sigma_2
		\cos^2(x_{l'}-\eta_j) \, 
		 	\frac{\omega_{1,j}}{L^{\alpha}} 
		+
		\sigma_1
		\cos^2(x_{l'}-\eta_j+k)
		 	\frac{\omega_{2,j+1}}{L^{\alpha}} 
		 +
		 O\left(
		 	\frac{1}{L^{2\alpha}}
		 \right)
		 \right)
		 \\
		 &&
		 =
		 \sigma_2^2
		\cos^2(x_l-\eta_j) \, \cos^2(x_{l'}-\eta_j)
		 	\frac{\omega_{1,j}^2}{L^{2\alpha}} 
		 \\
		 &&
		 \quad
		 +
		 \sigma_1^2
		\cos^2(x_{l}-\eta_j+k)\, \cos^2(x_{l'}-\eta_j+k)
		 	\frac{\omega_{2,j+1}^2}{L^{2\alpha}} 
		 +
		 B_j \omega_{1,j}\omega_{2,j+1}
		 +
		 O\left(
		 	\frac{1}{L^{3\alpha}}
		 \right),
	\end{eqnarray*}
	for some sequence $(B_j)_j$.
	We can see immediately that the limit \eqref{eq:general-diffusion} will be $0$ if $\alpha>\frac{1}{2}$. We consider $\alpha=\frac{1}{2}$ in the following. Using Lemma \ref{thm:trigonometric-limits-1},
	\begin{eqnarray*}
		&&
		\lim_{L\to\infty}
		\sum^{Lt}_{j=1}
		\esp\left[
			\cos^2(x_l-\eta_j) \, \cos^2(x_{l'}-\eta_j)
		 	\frac{\omega_{1,j}^2}{L^{2\alpha}} 
		\right]
		\\
		&& \quad
		=
		\lim_{L\to\infty}
		\frac{1}{4L}
		\sum^{Lt}_{j=1}
		\Big{(}
			1+\cos(2x_l-2\eta_j)\,\cos(2x_{l'}-2\eta_j)
		\Big{)}
		\\
		&& \quad
		=
		\lim_{L\to\infty}
		\frac{1}{8L}
		\sum^{Lt}_{j=1}
		\Big{(}
			2+\cos(2x_l-2x_{l'})+\cos(2x_{l}+2x_{l'}-4\eta_j)
		\Big{)}
		\\
		&& \quad
		=
		\frac{t}{4}
		+ \frac{t}{8} \cos(2x_l-2x_{l'}).
	\end{eqnarray*}
	Hence, for $l,l'\in\{1,\cdots,m\}$,
	\begin{eqnarray*}
		\lim_{L\to\infty}
		\int^t_0 a^L_{l,l'}(s,x)\, ds
		=
		\frac{\sigma^2 t}{4}
		\left(
			1 + \frac{1}{2}\cos(2x_l-2x_{l'})
		\right).
	\end{eqnarray*}
	Next,
	\begin{eqnarray*}
		&&
		\text{Re}\left\{\log\Big{(} 1 + \Gamma_j(\lambda_l,x_l) \Big{)}\right\}
		\text{Re}\left\{\log\Big{(} 1 + \Gamma_j(\lambda_{l'},x_{l'}) \Big{)}\right\}
		\\
		&&
		=
		\left(
		\frac{\sigma_2}{2}
		\sin(2x_l-2\eta_j)
		 	\frac{\omega_{1,j}}{L^{\alpha}} 
		+
		\frac{\sigma_1}{2}
		\sin(2x_l-2\eta_j+2k) 
		 	\frac{\omega_{2,j+1}}{L^{\alpha}}
		 +
		 O\left(
		 	\frac{1}{L^{2\alpha}}
		 \right)
		 \right)
		\\
		&&
		\quad
		\times
		\left(
		\frac{\sigma_2}{2}
		\sin(2x_{l'}-2\eta_j)
		 	\frac{\omega_{1,j}}{L^{\alpha}} 
		+
		\frac{\sigma_1}{2}
		\sin(2x_{l'}-2\eta_j+2k) 
		 	\frac{\omega_{2,j+1}}{L^{\alpha}}
		 +
		 O\left(
		 	\frac{1}{L^{2\alpha}}
		 \right)
		 \right)
		 \\
		 &&
		 =
		 \frac{\sigma_2^2}{4}
		 \sin(2x_{l}-2\eta_j) \, \sin(2x_{l'}-2\eta_j)
		 \frac{\omega_{1,j}^2}{L^{2\alpha}} 
		 \\
		 &&
		 \quad
		 +
		 \frac{\sigma_1^2}{4}
		\sin(2x_{l}-2\eta_j+2k) \, \sin(2x_{l'}-2\eta_j+2k) 
		\frac{\omega_{2,j+1}^2}{L^{2\alpha}}
		+
		C_j \omega_{1,j}\omega_{2,j+1}
		+
		 O\left(
		 	\frac{1}{L^{3\alpha}}
		 \right),
	\end{eqnarray*}
	for some sequence $(C_j)_j$.
	Now,
	\begin{eqnarray*}
		&&
		\lim_{L\to\infty}
		\sum^{Lt}_{j=1}
		\esp\left[
			\sin(2x_{l}-2\eta_j) \, \sin(2x_{l'}-2\eta_j)
		 	\frac{\omega_{1,j}^2}{L^{2\alpha}}
		\right]
		\\
		&&
		\quad
		=
		\lim_{L\to\infty}
		\frac{1}{2L}
		\sum^{Lt}_{j=1}
		\Big{(}
			\cos(2x_l-2x_{l'})
			-
			\cos(2x_l+2x_{l'}-4\eta_j)
		\Big{)}
		=
		\frac{t}{2} \cos(2x_l-2x_{l'}).
	\end{eqnarray*}
	Hence, for $l,l'\in\{1,\cdots,m\}$,
	\begin{eqnarray*}
		\lim_{L\to\infty}
		\int^t_0 a^L_{m+l,m+l'}(s,x)\, ds
		=
		\frac{\sigma^2 t}{8}\cos(2x_l-2x_{l'}).
	\end{eqnarray*}
	Finally,
	\begin{eqnarray*}
		&&
		\text{Im}\left\{\log\Big{(} 1 + \Gamma_j(\lambda_l,x_l) \Big{)}\right\}
		\text{Re}\left\{\log\Big{(} 1 + \Gamma_j(\lambda_{l'},x_{l'}) \Big{)}\right\}
		\\
		&&
		=
		\left(
		\sigma_2
		\cos^2(x_l-\eta_j) \, 
		 	\frac{\omega_{1,j}}{L^{\alpha}} 
		+
		\sigma_1
		\cos^2(x_l-\eta_j+k)
		 	\frac{\omega_{2,j+1}}{L^{\alpha}}
		 +
		 O\left(
		 	\frac{1}{L^{2\alpha}}
		 \right)
		 \right)
		 \\
		 &&
		 \quad
		\times
		\left(
		\frac{\sigma_2}{2}
		\sin(2x_{l'}-2\eta_j)
		 	\frac{\omega_{1,j}}{L^{\alpha}} 
		+
		\frac{\sigma_1}{2}
		\sin(2x_{l'}-2\eta_j+2k) 
		 	\frac{\omega_{2,j+1}}{L^{\alpha}}
		 +
		 O\left(
		 	\frac{1}{L^{2\alpha}}
		 \right)
		 \right)
		 \\
		 &&
		 =
		 \frac{\sigma_2^2}{2}
		 \cos^2(x_l-\eta_j) \, \sin(2x_{l'}-2\eta_j)
		 \frac{\omega_{1,j}^2}{L^{2\alpha}}
		 \\
		 &&
		 \quad
		 +
		 \frac{\sigma_1^2}{2}
		 \cos^2(x_l-\eta_j+k) \, \sin(2x_{l'}-2\eta_j+2k)
		 \frac{\omega_{2,j+1}^2}{L^{2\alpha}}
		 +
		 D_j \omega_{1,j}\omega_{2,j+1}
		 +
		 O\left(
		 	\frac{1}{L^{3\alpha}}
		 \right),
	\end{eqnarray*}
	for some sequence $(D_j)_j$.
	Now,
	\begin{eqnarray*}
		&&
		\lim_{L\to\infty}
		\sum^{Lt}_{j=1}
		\esp\left[
			 \cos^2(x_l-\eta_j) \, \sin(2x_{l'}-2\eta_j)
		 	\frac{\omega_{1,j}^2}{L^{2\alpha}}
		\right]
		\\
		&&
		\quad
		=
		\lim_{L\to\infty}
		\frac{1}{2L}
		\sum^{Lt}_{j=1}
		\Big{(}
			1 + \cos(2x_l-2\eta_j)
		\Big{)}
		\,
		 \sin(2x_{l'}-2\eta_j)
		 \\
		 &&
		 \quad
		=
		\lim_{L\to\infty}
		\frac{1}{4L}
		\sum^{Lt}_{j=1}
		\Big{(}
			\sin(2x_l+2x_{l'}-4\eta_j)
			-
			\sin(2x_l-2x_{l'})
		\Big{)}
		\\
		&&
		\quad
		=
		-\frac{t}{4}\sin(2x_l-2x_{l'}).
	\end{eqnarray*}
	Hence, for $l,l'\in\{1,\cdots,m\}$,
	\begin{eqnarray*}
		\lim_{L\to\infty}
		\int^t_0 a^L_{l,m+l'}(s,x)\, ds
		=
		-\frac{\sigma^2 t}{8}\sin(2x_l-2x_{l'}).
	\end{eqnarray*}
	We are left with the terms $l'=2m+1$ and $l'=2m+2$. 
	For $l,l'\in\{1,\cdots,m\}$,
	\begin{eqnarray*}
		\lim_{L\to\infty}
		\int^t_0 a^L_{l,2m+1}(s,x)\, ds
		&=&
		-
		\lim_{L\to\infty}
		\sum^{Lt}_{j=1}
		\esp\left[
			\text{Im}\, 
			\log\Big{(}
				1 + \Gamma_j(\lambda_l,x_l)
			\Big{)}
			\frac{\omega_{1,j}}{L^{\alpha}}
		\right]
		\\
		&=&
		-\sigma_2
		\lim_{L\to\infty}
		\frac{1}{L}
		\sum^{Lt}_{j=1}
		\cos^2(x_l-\eta_j)
		=
		-\frac{\sigma_2 t}{2},
	\end{eqnarray*}
	and, similarly,
	\begin{eqnarray*}
		\lim_{L\to\infty}
		\int^t_0 a^L_{l,2m+2}(s,x)\, ds
		=
		-\frac{\sigma_1 t}{2}.
	\end{eqnarray*}
	Next,
	\begin{eqnarray*}
		\lim_{L\to\infty}
		\int^t_0 a^L_{m+l,2m+1}(s,x)\, ds
		&=&
		-
		\lim_{L\to\infty}
		\sum^{Lt}_{j=1}
		\esp\left[
			\text{Re}\, 
			\log\Big{(}
				1 + \Gamma_j(\lambda_l,x_l)
			\Big{)}
			\frac{\omega_{1,j}}{L^{\alpha}}
		\right]
		\\
		&=&
		-\sigma_2
		\lim_{L\to\infty}
		\frac{1}{L}
		\sum^{Lt}_{j=1}
		\sin(2x_l-2\eta_j)
		=
		0,
	\end{eqnarray*}
	and
	\begin{eqnarray*}
		\lim_{L\to\infty}
		\int^t_0 a^L_{m+l,2m+1}(s,x)\, ds
		=
		0.
	\end{eqnarray*}
	Finally,
	\begin{eqnarray*}
		\lim_{L\to\infty}
		\int^t_0 a^L_{2m+1,2m+2}(s,x)\, ds
		&=&
		\lim_{L\to\infty}
		\int^t_0 a^L_{2m+2,2m+1}(s,x)\, ds
		=
		0,
		\\
		\lim_{L\to\infty}
		\int^t_0 a^L_{2m+1,2m+1}(s,x)\, ds
		&=&
		\lim_{L\to\infty}
		\int^t_0 a^L_{2m+2,2m+2}(s,x)\, ds
		=t.
	\end{eqnarray*}
	We have then verified the first condition on Lemma \ref{thm:general-convergence-finite-dimensional}. The second condition is easily verified by an inspection at \eqref{eq:log-approximation} and the third one follows from the fact that $\alpha \geq \frac12$.

	Let us identify the limiting SDE. We start writing the correlation matrix for $m=2$ (the general case being only notationally more involved), which is $a(t,x)=a(x)$ given by
	\begin{eqnarray*}
		\begin{pmatrix}
			\frac{3\sigma^2}{8}
			&
			\frac{\sigma^2}{4}\left( 1 + \frac{1}{2}\cos(2x_1-2x_2) \right)
			&
			0
			&
			-\frac{\sigma^2}{4}\sin(2x_1-2x_2)
			&
			-\frac{\sigma_2}{2}
			&
			-\frac{\sigma_1}{2}
			\\
			\\
			%
			\frac{\sigma^2}{4}\left( 1 + \frac{1}{2}\cos(2x_1-2x_2) \right)
			&
			\frac{3\sigma^2}{8}
			&
			-\frac{\sigma^2}{4}\sin(2x_2-2x_1)
			&
			0
			&
			-\frac{\sigma_2}{2}
			&
			-\frac{\sigma_1}{2}
			\\
			\\
			%
			0
			&
			-\frac{\sigma^2 }{8} \sin(2x_2-2x_1)
			&
			\frac{\sigma^2}{8}			
			&
			\frac{\sigma^2}{8}	\cos(2x_1-2x_2)
			&
			0
			&
			0
			\\
			\\
			%
			-\frac{\sigma^2 }{8} \sin(2x_1-2x_2)
			&
			0
			&
			\frac{\sigma^2}{8}	\cos(2x_1-2x_2)
			&
			\frac{\sigma^2}{8}	
			&
			0
			&
			0
			\\
			\\
			%
			-\frac{\sigma_2 }{2}
			&
			-\frac{\sigma_2 }{2}
			&
			0
			&
			0
			&
			1
			&
			0
			\\
			\\
			%
			-\frac{\sigma_1}{2}
			&
			-\frac{\sigma_1 }{2}
			&
			0
			&
			0
			&
			0
			&
			1
		\end{pmatrix}.
	\end{eqnarray*}
	This can be written as $g(x)g(x)^T$ with $g(x)$ given by
	\begin{eqnarray*}
		\begin{pmatrix}
			%
			\frac{\sigma}{2\sqrt{2}}\sin(2x_1)
			&
			\frac{\sigma}{2\sqrt{2}}\cos(2x_1)
			&
			0
			&
			0
			&
			-\frac{\sigma_2}{2}
			&
			-\frac{\sigma_1}{2}
			\\
			\\
			%
			\frac{\sigma}{2\sqrt{2}}\sin(2x_2)
			&
			\frac{\sigma}{2\sqrt{2}}\cos(2x_2)
			&
			0
			&
			0
			&
			-\frac{\sigma_2}{2}
			&
			-\frac{\sigma_1}{2}
			\\
			\\
			%
			\frac{\sigma}{2\sqrt{2}}\cos(2x_1)
			&
			-\frac{\sigma}{2\sqrt{2}}\sin(2x_1)
			&
			0
			&
			0
			&
			0
			&
			0
			\\
			\\
			%
			\frac{\sigma}{2\sqrt{2}}\cos(2x_2)
			&
			-\frac{\sigma}{2\sqrt{2}}\sin(2x_2)
			&
			0
			&
			0
			&
			0
			&
			0
			\\
			\\
			%
			0
			&
			0
			&
			0
			&
			0
			&
			1
			&
			0
			\\
			\\
			%
			0
			&
			0
			&
			0
			&
			0
			&
			0
			&
			1
		\end{pmatrix}
	\end{eqnarray*}
	By Lemma \ref{thm:general-convergence-finite-dimensional}, we obtain the joint convergence
	\begin{eqnarray*}
		\vartheta^{\lambda_i}_L \Rightarrow \vartheta^{\lambda_i},
		\quad
		r^{\lambda_l}_L \Rightarrow r^{\lambda_l},
		\quad l=1,\cdots,m,
	\end{eqnarray*}
	in law in the topology of uniform convergence in $[0,1]$
	to the solution of the system of SDEs 
	\begin{eqnarray*}
		d\vartheta^{\lambda_l}
		&=&
		\lambda_l dt
		+
		\frac{\sigma}{\sqrt{2}}\sin(\vartheta^{\lambda_l})\, dB_1
		+
		\frac{\sigma}{\sqrt{2}}\cos(\vartheta^{\lambda_l})\, dB_2
		+
		\sigma dW,
		\\
		dr^{\lambda_l}
		&=&
		\frac{\sigma^2}{8}dt
		+
		\frac{\sigma}{2\sqrt{2}}\cos(\vartheta^{\lambda_l})\, dB_1
		-
		\frac{\sigma}{2\sqrt{2}}\sin(\vartheta^{\lambda_l})\, dB_2,
	\end{eqnarray*}
	where $B_1,B_2,W$ are three independent standard Brownian motion.
	With the complex Brownian motion $B=\frac{B_2+iB_1}{\sqrt{2}}$, the above becomes
	\begin{eqnarray*}
		d\vartheta^{\lambda_l}
		&=&
		\lambda_l dt
		+
		\sigma
		\text{Re}\left\{
			e^{-i\vartheta^{\lambda}}dB
		\right\}
		+
		\sigma dW
		\\
		dr^{\lambda_l}
		&=&
		\frac{\sigma^2}{8}dt
		+
		\frac{\sigma}{2}
		\text{Im}\left\{
			e^{-i\vartheta^{\lambda_l}}dB
		\right\}.
	\end{eqnarray*}
	We also obtain the convergence of $\lambda\mapsto\vartheta^{\lambda}_L(1)$ in the sense of finite-dimensional distributions. By Lemma \ref{thm:monotonicity-phases} below, these are increasing functions so that the tension in the topology of uniform convergence on compact sets follows from Lemma \ref{thm:second-dini}.
\end{proof}

\begin{remark}
Note that we have the equality in law
\begin{eqnarray*}
	dr^{\lambda}
	&=&
	\frac{\sigma^2}{8}dt
	+
	\frac{\sigma}{2\sqrt{2}}
	dB,
	\quad
	r^{\lambda}(0)=0.
\end{eqnarray*}
Hence, $|\zeta_L|^2$ converges to
\begin{eqnarray*}
	e^{2r^{\lambda}(t)}
	=
	\exp\left\{
		\frac{\sigma^2 t}{4}
		+
		\frac{\sigma}{\sqrt{2}}B_t
	\right\}.
\end{eqnarray*}
From the point of view of the eigenfunctions, this yields the limit
\begin{eqnarray*}
	\left|
		\psi(t)
	\right|^2
	=
	\frac{
		\exp\left\{
		\frac{\sigma^2 t}{4}
		+
		\frac{\sigma}{\sqrt{2}}B_t
		\right\}
	}
	{
		\int^t_0
		\exp\left\{
		\frac{\sigma^2 t}{4}
		+
		\frac{\sigma}{\sqrt{2}}B_t
		\right\}
		ds
	}.
\end{eqnarray*}
\end{remark}

The following lemma contains the monotonicity used at the end of the proof above.
\begin{lemma}\label{thm:monotonicity-phases}
	The random function $\lambda \mapsto \theta^{\lambda}_L(L)$ is increasing for each $L\geq 1$. 
\end{lemma}
\begin{proof}
	Denote $\theta^{\lambda}_L=\theta^{\lambda}_L(L)$, $\bar{\theta}^{\lambda}_L = \theta^{\lambda}_L - (2L+1)k$, and
	let $G_L(\lambda)=G_L(E+\tfrac{\lambda}{\rho L}; L,-;L,-)$.
	From \cite[Lemma 8.2]{BMT01}, the definition of the Pr\"ufer transform and the matrix $\ppru_n$ from \eqref{eq:change-of-basis}, we have the identities
	\begin{eqnarray*}
		G_L(\lambda)
		&=&
		-\frac{x^-_L}{x^+_L}
		=
		\frac{\sqrt{-p_1}}{\sqrt{p_2}} \frac{\cos(\bar{\theta}_L^{\lambda}+k)}{\cos(\bar{\theta}_L^{\lambda})}
		\\
		&=&
		\frac{\sqrt{-p_1}}{\sqrt{p_2}}
		\left( \cos k - \sin k \tan (\bar{\theta}^{\lambda}_L) \right).
	\end{eqnarray*}
	Now, let $(\lambda_j)_j$ be an enumeration of the eigenvalues of the operator $L\rho(D_L-E)$. Then, there exists positive weights $a_j$ such that
	\begin{eqnarray*}
		G_L(\lambda)
		=
		\sum_j \frac{a_j}{\lambda - \lambda_j}.
	\end{eqnarray*}
	The function $\lambda \mapsto G_L(\lambda)$ is then increasing in each interval $(\lambda_{j-1},\lambda_j)$ with $\displaystyle\lim_{\lambda\to\lambda_j^+}G_L(\lambda)=-\infty$ and $\displaystyle\lim_{\lambda\to\lambda_j^-}G_L(\lambda)=\infty$. As $\sin k <0$, the function $\lambda\mapsto \tan (\bar{\theta}^{\lambda}_L)$ satisfies the same property. As $\lambda \mapsto \bar{\theta}^{\lambda}_n$ is continuous it is therefore globally increasing and the result follows.
\end{proof}


\section{Scaling limit for Model II}


We will now sketch the proof of Theorem \ref{thm:scaling-model-II}.
This time, we consider the recursion
\begin{eqnarray*}
	\zeta^{\lambda}_L(j+1)
	&=&
	\Big{(}
	1
	- i\sigma_2
	\cos \bar\theta^{\lambda}_L(j) \, 
	e^{-i\bar\theta^{\lambda}_L(j)}
	  \left(
		 \frac{\omega_{1,j}}{(L-j)^{\alpha}} - \frac{\lambda}{\rho L}
	\right)
	\\
	&&
	\phantom{blabla}
	- i\sigma_1
	\cos( \bar\theta^{\lambda}_L(j)-k)
	\, e^{-i (\bar\theta^{\lambda}_L(j)-k)}
	\left(
		 \frac{\omega_{2,j+1}}{(L-j)^{\alpha}} - \frac{\lambda}{\rho L}
	\right)
	\\
	\nonumber
	&&
	\phantom{blabla}
	-i \sqrt{\sigma_1\sigma_2}
	\cos \bar\theta^{\lambda}_L(j) \,
	e^{-i (\bar \theta^{\lambda}_L(j)-k)}
	\left(
		 \frac{\omega_{1,j}}{(L-j)^{\alpha}} - \frac{\lambda}{\rho L}
	\right)
	\left(
		 \frac{\omega_{2,j+1}}{(L-j)^{\alpha}} - \frac{\lambda}{\rho L}
	\right)
	\Big{)}\,
	\zeta^{\lambda}_L(j)
	\\
	&=:&
	\Big{(}
		1 +\Gamma_j(\lambda,\theta_L(j))
	\Big{)} 
	\, \zeta^{\lambda}_L(j),
\end{eqnarray*}
with $R^{\lambda}_L(0)=1$ and $\theta^{\lambda}_L(j)=0$, where we recall that we conveniently reversed the envelope on the environment.

An inspection of the proof for Model I shows that Model II can be treated in the exact same way, except that each application of Lemma \ref{thm:trigonometric-limits-1} has to be replaced by the corresponding use of Lemma \ref{thm:trigonometric-limits-2}. For instance, in the computation of
\begin{eqnarray*}
	\lim_{L\to\infty}
		\int^t_0 a^L_{l,l'}(s,x)\, ds
\end{eqnarray*}
for $l,l'\in\{1,\cdots,m\}$ and $t\in[0,1)$, we are lead to compute the limit
\begin{eqnarray*}
		&&
		\lim_{L\to\infty}
		\sum^{Lt}_{j=1}
		\esp\left[
			\cos^2(x_l-\eta_j) \, \cos^2(x_{l'}-\eta_j)
		 	\frac{\omega_{1,j}^2}{(L-j)^{2\alpha}} 
		\right]
		\\
		&& \quad
		=
		\lim_{L\to\infty}
		\frac{1}{4}
		\sum^{Lt}_{j=1}
		\frac{
			1+\cos(2x_l-2\eta_j)\,\cos(2x_{l'}-2\eta_j)
		}
		{(L-j)^{2\alpha}}
		\\
		&& \quad
		=
		\lim_{L\to\infty}
		\frac{1}{8}
		\sum^{Lt}_{j=1}
		\frac{
			2+\cos(2x_l-2x_{l'})+\cos(2x_{l}+2x_{l'}-4\eta_j)
		}
		{(L-j)^{2\alpha}}
		\\
		&& \quad
		=
		\left(
			\frac{1}{4}
			+ \frac{1}{8} \cos(2x_l-2x_{l'})
		\right)
		\int^t_0 \frac{ds}{1-s},
	\end{eqnarray*}
	for $\alpha=\frac12$, while the limit is $0$ for $\alpha>\frac12$.
	Hence, for $l,l'\in\{1,\cdots,m\}$, $\alpha=1\frac12$ and $t\in[0,1)$,
	\begin{eqnarray*}
		\lim_{L\to\infty}
		\int^t_0 a^L_{l,l'}(s,x)\, ds
		=
		\frac{\sigma^2 }{4}
		\left(
			1 + \frac{1}{2}\cos(2x_l-2x_{l'})
		\right)
		\int^t_0 \frac{ds}{1-s}.
	\end{eqnarray*}
	All the terms $a^L_{l,l'}(s,x)$ can be treated similarly. Note that the treatment of the drift is identical.
	
	Lemma \ref{thm:general-convergence-finite-dimensional} then yields convergence in law to the claimed SDE, uniformly in each interval $[0,T]$ with $T\in[0,1)$.

We note that, in order to obtain the scaling limit of the spectrum, one needs to obtain the limit of the relative phase given by
\begin{eqnarray}
	\label{eq:convergence-relative-phase}
	\vartheta^{\lambda}_L(L)-\vartheta^{0}_L(L)
	\Rightarrow
	g(\lambda)
	=
	\lim_{t\to1-}
	\left(
		\vartheta^{\lambda}(t)-\vartheta^{0}(0)
	\right).
\end{eqnarray}
Note that, by Lemma \ref{thm:monotonicity-phases}, the random function $\lambda\mapsto \vartheta^{\lambda}_L(L)-\vartheta^{0}_L(L)$ is non-decreasing and well-defined for each $L\geq 1$. The proof of \eqref{eq:convergence-relative-phase} is quite involved. It follows from the arguments outlined in \cite{KVV}, Section 6, which are detailed in  \cite{VV}, Section 6 and will not be reproduced here.


\section{Convergence of the point processes}


We illustrate the convergence of the point processes for Model 1 and $\alpha=\frac{1}{2}$. The case $\alpha>\frac12$ follows from similar arguments. Model II yields to additional technicalities. The interested reader can consult \cite{KVV}. 

The following result is in fact quite more general than stated.
\begin{proposition}
\label{thm:convergence-pp}
	Consider Model I with $\alpha=\frac{1}{2}$ and $k\neq -\frac{3\pi}{4}$.
	Then
	the point process $$\Upsilon_L - 2[(2L-1)k]-\pi,$$ 
	converges to $ {\bf Sch}_{\sigma^2}$ in the vague convergence topology where $\sigma^2 = \frac{p_1(E)^2+p_2(E)^2}{\sin^2(2k)}$.
\end{proposition}
\begin{proof}
	Recall that
	\begin{eqnarray*}
		\Upsilon_L
		=
		\left\{
			\lambda:\, \vtl_L(L) -[2\eta_L] -\pi  \in 2 \pi \Z
		\right\},
	\end{eqnarray*}
	with $\eta_L=(2L-1)k$,
	and recall that the increasing processes $(\lambda\mapsto\vtl_L(L))$ converge in law to $(\lambda\mapsto\vtl(1))$ in the topology of uniform convergence in compact set. Using Skorohod's representation theorem, we can assume that this convergence holds almost surely on a possibly different probability space.
	
	Assume that $[\eta_L]\to\eta$ and fix a realization of the processes such that $\vartheta_L(\cdot):=\vartheta^{\cdot}_L(L)$ converges to $\vartheta({\cdot}):=\vartheta^{\cdot}$, uniformly on compact sets.
	Let
	\begin{eqnarray*}
		\Upsilon 
		=
		\left\{
			\lambda:\, \vtl(1)-2\eta -\pi  \in 2 \pi \Z
		\right\}.
	\end{eqnarray*}
	As $\vartheta^{\cdot}_L$ and $\vartheta^{\cdot}$ are continuous and increasing, $\Upsilon_L$ and $\Upsilon$ are discrete countable sets. Furthermore, as
	\begin{eqnarray*}
		\vartheta^{\lambda}(t)
		&\stackrel{\bf{(d)}}{=}&
		\vartheta^{\lambda/\sigma^2}(\sigma^2 t),
	\end{eqnarray*}
	we have that
	\begin{eqnarray*}
		\Upsilon
		&\stackrel{\bf{(d)}}{=}&
		\left\{
			\lambda:\, \varphi^{\lambda/\sigma^2}(\sigma^2) -2\eta -\pi  \in 2 \pi \Z
		\right\}.
	\end{eqnarray*}
	Now, from \cite[Lemma 19]{KVV}, we know that
	\begin{eqnarray*}
		\varphi^{\lambda-\theta}(t) + \theta t
		&\stackrel{\bf{(d)}}{=}&
		\varphi^{\lambda}(t),
	\end{eqnarray*}
	which implies that
	\begin{eqnarray*}
		\textbf{Sch}_{\tau} + \theta
		&\stackrel{\bf{(d)}}{=}&
		\{\lambda:\, \varphi^{\lambda/\tau}(\tau) \in \theta + 2\pi\Z\}.
	\end{eqnarray*}
	Hence,\begin{eqnarray*}
		\Upsilon
		&\stackrel{\bf{(d)}}{=}&
		\textbf{Sch}_{\sigma^2}+2\eta+\pi.
	\end{eqnarray*}
	Label $\Upsilon = (\lambda_j)_j$ in increasing order.
	 Let $-\infty<a<b<\infty$ and let $M$ and $N$ be such that $(\lambda_j)^N_{j=M} = \Upsilon \cap [a,b]$.
	Let $\varepsilon\in(0,1)$ and let $L$ be large enough so that $\displaystyle\sup_{\lambda \in [a-1,b+1]} |\vartheta(\lambda)-\vartheta_L(\lambda)|<\varepsilon$. 
	Then, there exists a function $\delta$ with $\displaystyle\lim_{\varepsilon\to0} \delta(\varepsilon)=0$ such that, for each $j=M,\cdots,N$, there exists a unique $\lambda^L_j\in (\lambda_j-\delta(\varepsilon),\lambda_j+\delta(\varepsilon)) \cap \Upsilon_L$. This also shows that $\displaystyle\lim_{L\to\infty}\lambda^L_j = \lambda_j$ for each $j=1,\cdots,M$. Furthermore, for $\varepsilon$ small enough, there are no other points in $\Upsilon_L\cap[a-\delta(\varepsilon),b+\delta(\varepsilon)]$.
	
	Let $\mu$ and $\mu_L$ be the counting measures of $\Upsilon$ and $\Upsilon_L$ respectively. The above shows that for any non-negative, continuous and compactly supported function $f:\R\to\R$,
	\begin{eqnarray}
		\lim_{L\to\infty} \int f \, d\mu_L = \int f\, d\mu,
	\end{eqnarray}
	almost surely. By the dominated convergence theorem, it then holds that
	\begin{eqnarray}
		\lim_{L\to\infty}\esp\left[ e^{-\int f \, d\mu_L}\right]
		=
		\esp\left[ e^{-\int f \, d\mu}\right].
	\end{eqnarray}
	In other words, 
	\begin{eqnarray}
		\Upsilon_L \Rightarrow \textbf{Sch}_{\sigma^2} + 2\eta + \pi,
	\end{eqnarray}
	in the topology of vague convergence. 
	So far, as we assumed that $\eta_L\to\eta$, this limit is only subsequential. But, the result follows by noticing that $2[\eta_L]-[2\eta_L]\in\{0,2\pi\}$.	
\end{proof}


\appendix


\section{Lemmas about convergence of stochastic processes}


The following corresponds to \cite[Proposition 27]{KVV}. We reproduce it here for the convenience of the reader.
\begin{lemma}\label{thm:general-convergence-finite-dimensional}
	Consider a family of $\R^d$-valued Markov chains
	\begin{eqnarray*}
		\{X^n_j:\, j=1,\cdots, nT\},\quad n\geq 1.
	\end{eqnarray*}
	Let $Y^n_j(x)$ be distributed as the increment $X^n_{j+1}-X^n_j$ conditioned on $X^n_j=x$ and let
	\begin{eqnarray*}
		b^n(t,x) = n \esp\left[ Y^n_{nt}(x)\right],\quad a^n(t,x)_{ij} = n \esp\left[ Y^n_{nt}(x)_i Y^n_{nt}(x)_j\right].
	\end{eqnarray*}
	Assume that:
	\begin{enumerate}
		\item there exists two smooth functions $a:[0,T]\times \R^d \to M_d^{\text{sym}}(\R^d)$ and  $
	b:[0,T]\times \R^d \to M_d(\R^d)$ such that, for every $R>0$,
	\begin{eqnarray*}
		&&
		\sup_{\substack{0\leq t \leq T \\ |x|\leq R}}
		\left|
			\int^t_0 \left(
				a^n(s,x)-a(s,x)
			\right)\, ds
		\right|
		\\
		&&
		+
		\sup_{\substack{0\leq t \leq T \\ |x|\leq R}}
		\left|
			\int^t_0 \left(
				b^n(s,x)-b(s,x)
			\right)\, ds
		\right|
		\to
		0,
	\end{eqnarray*}
	as $n\to\infty$,
	
	\vspace{1ex}
	
	\item for every $R>0$, there exists a constant $C=C(R)$ such that
	\begin{eqnarray*}
		\| a^n(t,x)-a^n(t,y)\|
		+
		\| b^n(t,x)-b^n(t,y)\|
		\leq C \| x-y\|,
	\end{eqnarray*}
	for all $n\geq 1$, $t\in[0,T]$ and $\|x\|, \, \|y\| \leq R$ and the same inequality holds for $a$ and $b$,
	
	\vspace{1ex}
	
	\item for every $R>0$, there exists a constant $K=K(R)$ such that
	\begin{eqnarray*}
		\sup_{\substack{0\leq j \leq n \\ |x|\leq R}}
		\esp\left[
			\left\|
				Y^n_j(x)
			\right\|^3
		\right]
		\leq
		K n^{-3/2},
	\end{eqnarray*}
	
	\vspace{1ex}
	
	\item $X^n_0 \to X_0$ in distribution, with $\esp[X_0^2]<\infty$.
	\end{enumerate}
	
	\vspace{1ex}
	
	Then, $\{X^n_{tn}:\, t\in[0,T]\}$ converges weakly in $\mathcal{D}[0,T]$ to the unique solution of
	\begin{eqnarray*}
		dX_t = b(t,X_t)dt + g(t,X_t) dB_t,
	\end{eqnarray*}
	where $B$ is a $d$-dimensional Brownian motion and $g$ is a matrix-valued function such that $a = g g^T$.
\end{lemma}
The following is a stochastic version of the `second' Dini's theorem:
\begin{lemma}\label{thm:second-dini}
	Let $(X_n)_n$ be a sequence of monotone process converging to a process $X$ in the sense of finite dimensional distributions. Then, $(X_n)_n$ is tight in the topology of uniform convergence over compact sets. 
\end{lemma}
\begin{proof}
	Assume each $X_n$ is increasing. It is enough to show tightness in an interval $[0,T]$. Furthermore, it is enough to show that:
	\begin{enumerate}
		\item For all $\varepsilon>0$, there exists $A>0$ such that
			\begin{eqnarray*}
				\limsup_n\p\left[|X_n(0)| \geq A\right]<\epsilon.
			\end{eqnarray*}
			
		\item For all $\eta,\varepsilon>0$, there is a $\delta>0$ such that
			\begin{eqnarray*}
					\limsup_n \p\left[ \omega(X_n,\delta)>\epsilon\right] < \eta,
			\end{eqnarray*}
			where $\omega(f,\delta)=\sup\{|f(t)-f(s)|:\, 0 \leq s,t \leq T,\, |s-t|<\delta \}$ is the modulus of continuity.
	\end{enumerate}
	For the first point, choose $A$ such that $\p[|X(0)|\geq A]<\varepsilon$. Then, by Portemanteau's theorem,
	\begin{eqnarray*}
		\limsup_n \p\left[|X_n(0)| \geq A\right] \leq \p\left[|X(0)| \geq A\right]<\epsilon.
	\end{eqnarray*}
	Next, given $\eta,\varepsilon>0$, choose $\delta>0$ such that $\p\left[ X(\theta,2\delta)>\epsilon\right]<\eta$. Let $N$ be such that $N\delta \leq T < (N+1)\delta$ and define $t_j=j\delta$ for $j=0,\cdots, N$, $t_{-1}=0$ and $t_{N+1}=T$. Then, by monotonicity of $\theta$ and each $X_n$,
	\begin{eqnarray*}
		\limsup_n \p\left[ \omega(X_n,\delta)>\epsilon\right]
		&\leq&
		\limsup_n \p\left[ \bigcup^{N+1}_{j=1} \left\{ |X_n(t_j)-X_n(t_{j-2})| \geq \varepsilon\right\}\right]\\
		&\leq&
		\p\left[ \bigcup^{N+1}_{j=1} \left\{ |X(t_j)-X(t_{j-2})| \geq \varepsilon\right\}\right]\\
		&\leq&
		\p\left[ \omega(X,2\delta)>\epsilon\right]<\eta.
	\end{eqnarray*}
\end{proof}



\section{Limits of sums of trigonometric functions}

The following limits were used to obtain the scaling limit of the Pr\"ufer transform for Model I.
\begin{lemma}\label{thm:trigonometric-limits-1}
	Let $\eta_j=(2j-1)k$ for $k\in(-\pi,-\frac{\pi}{2})$ with $k\neq -\frac{3\pi}{2}$. Then, for all $x\in\R$, $t>0$ and $c\in\{1,2,4\}$,
	\begin{eqnarray*}
		\lim_{L\to\infty} 
		\frac{1}{L}
		\sum_{j=1}^{Lt}
		\sin(x-\eta_j)
		&=&
		\lim_{L\to\infty} 
		\frac{1}{L}
		\sum_{j=1}^{Lt}
		\cos(x-\eta_j)
		=
		0,
		\\
		\lim_{L\to\infty} 
		\frac{1}{L}
		\sum_{j=1}^{Lt}
		\sin^2(x-\eta_j)
		&=&
		\lim_{L\to\infty} 
		\frac{1}{L}
		\sum_{j=1}^{Lt}
		\cos^2(x-\eta_j)
		=
		\frac{t}{2},
		\\
		\lim_{L\to\infty} 
		\frac{1}{L}
		\sum_{j=1}^{Lt}
		\sin^4(x-\eta_j)
		&=&
		\lim_{L\to\infty} 
		\frac{1}{L}
		\sum_{j=1}^{Lt}
		\cos^4(x-\eta_j)
		=
		\frac{3t}{8},
		\\
		\lim_{L\to\infty} 
		\frac{1}{L}
		\sum_{j=1}^{Lt}
		\sin^2(x-\eta_j)
		\cos^2(x-\eta_j)
		&=&
		\frac{t}{8},
		\\
		\lim_{L\to\infty} 
		\frac{1}{L}
		\sum_{j=1}^{Lt}
		\cos^3(x-\eta_j)
		\sin(x-\eta_j)
		&=&
		0.
	\end{eqnarray*}
\end{lemma}
\begin{proof}
	The proof follows from the observation that
	\begin{eqnarray*}
		\sup_M
		\left(
			\left|
				\sum_{j=1}^{M}
			\sin(c(x-\eta_j))
			\right|
			+
			\left|
				\sum_{j=1}^{M}
			\cos(c(x-\eta_j))
			\right|
		\right)
		<
		\infty,
	\end{eqnarray*}
	for $c\in\{1,2,3\}$. For instance,
	\begin{eqnarray*}
		\lim_{L\to\infty} 
		\frac{1}{L}
		\sum_{j=1}^{Lt}
		\cos^2(x-\eta_j)
		=
		\lim_{L\to\infty} 
		\frac{1}{2L}
		\sum_{j=1}^{Lt}
		\Big{(}
			1+\cos(2(x-\eta_j))
		\Big{)}
		=
		\frac{t}{2}.
	\end{eqnarray*}
	The other cases follow similarly. 
\end{proof}
The following are the limits needed to treat Model II.
\begin{lemma}\label{thm:trigonometric-limits-2}
	Let $\eta_j=(2j-1)k$ for $k\in(-\pi,-\frac{\pi}{2})$ with $k\neq -\frac{3\pi}{2}$. Then, for all $x\in\R$, $t\in[0,1)$, $c\in\{1,2,4\}$ and $\alpha\geq \frac{1}{2}$,
	\begin{eqnarray*}
		\lim_{L\to\infty} 
		\sum_{j=1}^{Lt}
		\frac{\sin(x-\eta_j)}{(L-j)^{\alpha}}
		=
		\lim_{L\to\infty} 
		\sum_{j=1}^{Lt}
		\frac{\cos(x-\eta_j)}{(L-j)^{\alpha}}
		&=&
		0,
		\\
		&&
		\\
		\lim_{L\to\infty} 
		\sum_{j=1}^{Lt}
		\frac{\sin^2(x-\eta_j)}{(L-j)^{2\alpha}}
		=
		\lim_{L\to\infty} 
		\sum_{j=1}^{Lt}
		\frac{\cos^2(x-\eta_j)}{(L-j)^{2\alpha}}
		&=&
		\left\{
		\begin{array}{ll}
			0 & \alpha>\frac{1}{2},
			\\
			\\
			\frac{1}{2}\int^t_0 \frac{ds}{1-s} & \alpha=\frac{1}{2},
		\end{array}
		\right.
		\\
		&&
		\\
		\lim_{L\to\infty} 
		\sum_{j=1}^{Lt}
		\frac{\sin^4(x-\eta_j)}{(L-j)^{2\alpha}}
		=
		\lim_{L\to\infty} 
		\frac{1}{L}
		\sum_{j=1}^{Lt}
		\frac{\cos^4(x-\eta_j)}{(L-j)^{2\alpha}}
		&=&
		\left\{
		\begin{array}{ll}
			0 & \alpha>\frac{1}{2},
			\\
			\\
			\frac{3}{8}\int^t_0 \frac{ds}{1-s} & \alpha=\frac{1}{2},
		\end{array}
		\right.
		\\
		&&
		\\
		\lim_{L\to\infty} 
		\frac{1}{L}
		\sum_{j=1}^{Lt}
		\frac{\sin^2(x-\eta_j)
		\cos^2(x-\eta_j)}
		{(L-j)^{2\alpha}}
		&=&
		\left\{
		\begin{array}{ll}
			0 & \alpha>\frac{1}{2},
			\\
			\\
			\frac{1}{8}\int^t_0 \frac{ds}{1-s} & \alpha=\frac{1}{2},
		\end{array}
		\right.
		\\
		&&
		\\
		\lim_{L\to\infty} 
		\frac{1}{L}
		\sum_{j=1}^{Lt}
		\frac{\cos^3(x-\eta_j)
		\sin(x-\eta_j)}
		{(L-j)^{2\alpha}}
		&=&
		0.
	\end{eqnarray*}
\end{lemma}
\begin{proof}
	By means of adequate trigonometric identities, everything reduces to proving that
	\begin{eqnarray*}
		\lim_{L\to\infty} 
		\sum_{j=1}^{Lt}
		\frac{\sin(c(x-\eta_j))}{(L-j)^{\nu}}
		=
		0,
	\end{eqnarray*}
	for all $\nu\geq \frac{1}{2}$ and $c\in\{1,2,4\}$.
	But using that $t<1$, we have that
	\begin{eqnarray*}
		\limsup_{L\to\infty} 
		\sum_{j=1}^{Lt}
		\frac{\sin(c(x-\eta_j))}{(L-j)^{\nu}}
		\leq
		\lim_{L\to\infty} 
		\frac{1}{(1-t)^{\nu}L^{\nu}}
		\sum_{j=1}^{Lt}
		\sin(c(x-\eta_j))
		=
		0.
	\end{eqnarray*}
\end{proof}

\thanks{{\bf Acknowledgements.} 
We thanks O. Bourget for his many valuable comments on the model.
 G. Moreno is partially supported by the Chilean Fondecyt Grants 1211189 and 1211576. A. Taarabt is supported by the Chilean Fondecyt Grant 11190084.

\end{document}